\documentclass[journal]{IEEEtran}
\usepackage{amsmath,amsfonts}
\usepackage{algorithmic}
\usepackage{algorithm}
\usepackage{array}
\usepackage[caption=false,font=normalsize,labelfont=sf,textfont=sf]{subfig}
\usepackage{textcomp}
\usepackage{stfloats}
\usepackage{url}
\usepackage{verbatim}
\usepackage{graphicx}
\usepackage{cite}
\usepackage{amsthm}
\usepackage{booktabs}
\usepackage{multirow}
\usepackage{tablefootnote}
\usepackage{xcolor}
\usepackage[normalem]{ulem}

\newtheorem{theorem}{Theorem}
\newtheorem{assumption}{Assumption}
\newtheorem{remark}{Remark}
\newtheorem{lemma}{Lemma}
\newtheorem{proposition}{Proposition}

\begin{document}

\title{Lightweight Federated Learning in Mobile Edge Computing with Statistical and Device Heterogeneity Awareness}
\author{Jinghong Tan, ~Zhichen Zhang, ~Kun Guo*,~\IEEEmembership{Member,~IEEE}, ~Tsung-Hui Chang,~\IEEEmembership{Fellow,~IEEE}, ~Tony Q. S. Quek,~\IEEEmembership{Fellow,~IEEE}
\thanks{This work was supported in part by the National Research Foundation, Singapore and Infocomm Media Development Authority under its Future Communications Research \& Development Programme, in part by Shenzhen Science and Technology Program under Grant ZDSYS20230626091302006 and 
RCJC20210609104448114, in part by Guangdong Provincial Key Laboratory of Big Data Computing, in part by the National Natural Science Foundation of China, under Grant 62301222 and 62571191, and in part by Applied Basic Research Foundation of Yunnan Province under Grants 202301AT070198.} 
\thanks{Jinghong Tan and Zhichen Zhang are with the National Pilot School of Software, Yunnan University, Kunming 650500, China. (E-mails:  jinghong\_tansutd@hotmail.com; zhangzhichen@stu.ynu.edu.cn).}
\thanks{K. Guo is with the Shanghai Key Laboratory of Multidimensional Information Processing, School of Communications and Electronics Engineering, East China Normal University, Shanghai 200241, China. (email: kguo@cee.ecnu.edu.cn).}
\thanks{Tsung-Hui Chang is with the School of Artificial Intelligence, The Chinese University of Hong Kong, Shenzhen 518172, China, and also with the
Shenzhen Research Institute of Big Data, Shenzhen 518172, China (e-mail:tsunghui.chang@ieee.org).}
\thanks{Tony Q. S. Quek is with the Information Systems Technology and Design Pillar, Singapore University of Technology and Design, Singapore 487372, and also with the Yonsei Frontier Laboratory, Yonsei University, Seoul 03722, South Korea (e-mail: tonyquek@sutd.edu.sg).}
\thanks{\emph{Corresponding author: Kun Guo (e-mail:kguo@cee.ecnu.edu.cn)}.}
}

\maketitle

\begin{abstract}
Federated learning enables collaborative machine learning while preserving data privacy, but high communication and computation costs, exacerbated by statistical and device heterogeneity, limit its practicality in mobile edge computing. 
Existing compression methods like sparsification and pruning reduce per-round costs but may increase training rounds and thus the total training cost, especially under heterogeneous environments.
We propose a lightweight personalized FL framework built on parameter decoupling, which separates the model into shared and private subspaces, enabling us to uniquely apply gradient sparsification to the shared component and model pruning to the private one. 
This structural separation confines communication compression to global knowledge exchange and computation reduction to local personalization, protecting personalization quality while adapting to heterogeneous client resources. We theoretically analyze convergence under the combined effects of sparsification and pruning, revealing a sparsity-pruning trade-off that links to the iteration complexity. Guided by this analysis, we formulate a joint optimization that selects per-client sparsity and pruning rates and wireless bandwidth to reduce end-to-end training time. Simulation results demonstrate faster convergence and substantial reductions in overall communication and computation costs with negligible accuracy loss, validating the benefits of coordinated and resource-aware personalization in resource-constrained heterogeneous environments.

\end{abstract}


\section{Introduction}
\subsection{Background and Motivations}

\IEEEPARstart{A}{s} the Internet of Things and mobile edge computing (MEC) rapidly develop, massive amounts of data are generated at the network edge. Due to privacy constraints, centralized training is often impractical for such edge data. Federated learning (FL) \cite{pmlr-v54-mcmahan17a}, in which a central server coordinates local training across multiple clients while only exchanging model-related information, reconciles collaborative learning with data privacy and is thus regarded as a key technology for enabling distributed intelligence in edge scenarios such as intelligent transportation, mobile health, and smart cities.

Deploying FL in real-world MEC environments, however, faces two mutually interacting challenges. On the one hand, resources are constrained: mobile devices have limited uplink bandwidth, and energy \cite{9084352,tan2021robust}. On the other hand, clients are heterogeneous: user behavior and environments produce highly non-independent and identically distributed (non-IID) data distributions \cite{10304290,9609568}, while devices differ significantly in communication and computation capabilities. These dual constraints form the basic conditions that any personalized FL deployment in MEC must confront. Optimizing only one dimension is generally insufficient in practice, because in real scenarios, these two factors typically coexist and interact.

More importantly, these coexisting challenges exhibit a mutually reinforcing compound degradation effect, so that algorithm design is far more difficult than the simple superposition of single-issue solutions. Resource limitations restrict the feasibility of compensation strategies for data heterogeneity, for example, more sophisticated aggregation or personalization mechanisms, causing gradient conflict and model bias to accumulate during training. Simultaneously, beyond the inherent challenges of slower training and increased costs caused by data and device heterogeneity, their combined effect with resource limitations can also bias the global model towards resource-rich clients, thereby amplifying the degradation of generalization and fairness. Therefore, resource limits and the heterogeneity of data and devices  interact in complex ways in FL, and studying either factor in isolation cannot fully reveal the true performance frontier of the system. Together they form the core challenge for achieving efficient FL at the edge.

In resource-constrained and heterogeneous edge scenarios, prior work has reduced communication or computation overhead via quantization \cite{8917724,pmlr-v108-reisizadeh20a,9916128,10228970}, sparsification \cite{9488839,8889996,10.1145/3485730.3485929}, and pruning\footnote{Pruning can be classified by timing into static and dynamic pruning. Static pruning generally occurs post-training to accelerate inference, while dynamic pruning happens during training to speed up the training. In this paper, pruning generally refers to dynamic pruning.} \cite{9691274,9598845,9545941}, and numerous studies have proposed personalization strategies to mitigate the performance degradation caused by non-IID data \cite{9076082,9449957,NEURIPS2020_24389bfe,NEURIPS2020_f4f1f13c,li2019fedmd,NEURIPS2020_18df51b9,tan2024longterm}. However, a systematic gap remains in the literature: most methods either treat personalization merely as a tool to improve accuracy while overlooking its structural value for coordinated resource optimization, or treat compression techniques as single-round accelerators while ignoring their effects on total training rounds. Consequently, simultaneously addressing data and device heterogeneity with low time and resource costs remains largely unexplored.

Among various personalization strategies, parameter decoupling \cite{9743558,arivazhagan2019federated,liang2020think} splits the global FL model into a globally shared part and a locally private part. This approach mitigates statistical heterogeneity without introducing significant additional overhead, but it does not, by itself, alleviate problems caused by device heterogeneity. Motivated by this, in this paper, we elevate parameter decoupling to a structural mechanism for joint optimization of resource allocation, pruning, and sparsification, and we systematically characterize how the errors induced by sparsification and pruning couple and affect convergence and resource consumption. This combination is not a simple stacking of techniques. To address the complex challenges posed by resource-constrained MEC systems, we perform a systematic, structured design that integrates parameter decoupling, gradient sparsification, and model pruning.

Concretely, by exploiting the structural properties of parameter decoupling, we strictly confine gradient sparsification to the shared subspace, so that it only affects the efficiency of transmitting global knowledge; and confine model pruning to the personalization subspace, so that it only affects local computational efficiency. This function-level orthogonalization  protects both personalization performance and the consistency of global knowledge when adapting pruning and sparsification to heterogeneous device resources. Built on this unique structure, we reveal and model a previously under-explored “dual-coupling” relationship. First, gradient sparsification and model pruning,  albeit acting on different subspaces, become mutually constraining in the presence of limited system resources. Second, the training errors introduced by sparsification and pruning jointly affect the total number of iterations required for convergence. Therefore, the choices of pruning rate and sparsity rate directly determine the core trade-off between convergence speed and resource consumption — a topic that has been insufficiently explored in prior work.

To reveal the “dual-coupling” relationship, 
we provide an in-depth analysis of both the theoretical and system-level behavior of this proposed framework. First, from a convergence perspective, we incorporate the errors induced by gradient sparsification and model pruning into the parameter-decoupled global convergence analysis and derive an upper bound that explicitly contains sparsity and pruning rate terms, thereby quantifying their joint effect on iteration complexity and final error. Second, from a systems perspective, under heterogeneous communication and compute constraints, we formulate a joint optimization problem that aims to accelerate convergence. This problem treats each client’s sparsity rate, pruning rate, and wireless bandwidth allocation as optimization variables, and we design efficient algorithms to obtain high-quality solutions. Extensive simulations show that, under representative MEC heterogeneity, our method substantially reduces the total computation–communication cost and speeds up convergence while maintaining or closely matching baseline accuracy. These results confirm that the proposed joint optimization approach can find system-level resource–performance trade-offs under multi-dimensional constraints including data, compute, and communication heterogeneity. 

Based on the above, our main contributions are summarized as follows:
\begin{itemize}
    \item 
    We propose a novel lightweight personalized FL training method that elevates parameter decoupling from a personalization tool to a structural mechanism enabling function-level orthogonalization. In this design, gradient sparsification and model pruning are strictly confined to the shared and personalization subspaces, respectively. This mechanism preserves personalization performance and global knowledge consistency while reducing communication and computation costs and adapting to heterogeneous device resources.
    
    \item 
    Building on this structure, we identify and model a previously under-explored dual coupling between gradient sparsification and model pruning. Specifically, we show that the two techniques become mutually constraining in the presence of limited system resources. Furthermore, we prove that the training errors they introduce jointly affect the convergence efficiency, and we derive a new upper bound that explicitly quantifies their combined effect. These insights lead us to formulate a joint optimization problem of resource allocation, pruning, and sparsification, which achieves lightweight acceleration of convergence under heterogeneous constraints.  
    
    \item
    Extensive simulations reveal the joint impact of pruning and sparsification on the trade-off between convergence efficiency and resource consumption. The results further confirm that our method reduces total computation–communication cost and accelerates convergence with negligible accuracy loss, validating the effectiveness of our joint optimization approach in resource-constrained and heterogeneous MEC environments.  
\end{itemize}

\subsection{Related Works}
\subsubsection{Compression techniques for reducing communication and computation costs}

To address the high communication and computation costs in federated learning, existing studies primarily adopted model compression techniques (e.g., quantization, sparsification, and pruning) to reduce per-round communication or computation overhead. Among these methods, sparsification is a common and effective means of communication compression, allowing clients to upload only parts of the parameters or gradients. Compared with directly sparsifying weights, sparsifying gradients typically had less impact on training, and therefore gradient sparsification attracted more attention for reducing communication volume. Prior work reduced uplink data traffic via Top-k, random sampling and similar strategies, and often combined sparsification with quantization for further compression \cite{aji-heafield-2017-sparse, NEURIPS2018_b440509a,8889996,9488839}. For example, the work in \cite{9488839} proposed an adaptive sparsity control scheme based on Top-k to balance communication and computation energy among resource-heterogeneous clients, but it did not consider the effect of statistical heterogeneity; the work in \cite{8889996} combined gradient sparsification with ternary quantization in non-IID settings and achieved improved performance, yet it did not systematically relate sparsity design to device heterogeneity. Studies in sparsification ignored the heterogeneity either in device or data.
Moreover, since the information loss introduced by sparsification might have increased the number of rounds required for convergence \mbox{\cite{9488839,jiang2022model}}, existing studies in sparsification had limited improvement on the overall cost reduction and convergence acceleration, especially in resource- and device-heterogeneous MEC environments.

A different sparsification approach was to learn a sparse mask while keeping the model parameters unchanged during local training. FedMask \cite{10.1145/3485730.3485929} let clients learn binary sparse masks to derive personalized subnetworks from an initial model. Because only masks were communicated, this method had certain adaptability to heterogeneous communication conditions while addressing data heterogeneity; however, since the model’s structure and parameters were not adapted by each client’s data, its ability to adapt to compute heterogeneity was limited, so that the scale of the model was limited by the bottleneck device. More importantly, the quality of the initial model substantially affected the performance of the derived subnetworks.

Another common compression technique was model pruning, which accelerated training and inference by removing redundant weights. Unlike sparsification that often acted on gradients, pruning directly altered model structure and was therefore commonly used to reduce local computation and model size. In the federated learning context, some works attempted to jointly optimize pruning rates, device selection, and resource allocation to accommodate heterogeneous devices \cite{9598845}, but such studies often relied on relatively basic pruning strategies and did not fully account for statistical heterogeneity. To mitigate issues caused by non-IID data, the work in \cite{9545941} proposed hybrid structured and unstructured pruning strategies that identified client-appropriate subnetworks from the global model to reduce per-round communication and computation; this approach improved over baselines on non-IID data but did not explicitly address device heterogeneity, and the aggregation relying on overlapping parameters might have weakened aggregation signals under extreme pruning without rigorous convergence guarantees.

In summary, research on sparsification and pruning has made significant progress in reducing per-round communication or computation overhead, but there remains a clear gap in simultaneously accounting for statistical heterogeneity and device heterogeneity and in evaluating their combined impact on convergence rounds and total training cost.

\subsubsection{Personalized Federated Learning}

When client data are highly heterogeneous, using a single global model for all clients often results in severe gradient conflicts. Personalized federated learning alleviates this problem by allowing different clients to train models with differences, so that globally shared knowledge and client-specific information can be learned concurrently.

To this end, most personalization methods incurred additional learning overhead. For example, transfer learning or meta-learning based approaches \cite{9076082,9449957,NEURIPS2020_24389bfe,NEURIPS2020_f4f1f13c} typically alleviated performance degradation due to non-IID data but required extra local personalization steps beyond global training, which added computational cost. Knowledge-distillation based methods \mbox{\cite{li2019fedmd,NEURIPS2020_18df51b9,Shen_2022_CVPR}} avoided data privacy leakage by using intermediate activation transmission, pseudo-public data, or shared generators; such approaches either increased computation (e.g., training a generator or performing iterative teacher-student distillation locally) or incurred additional communication overhead (e.g., transmitting activations to serve as a virtual public dataset), and thus were limited in privacy- or resource-constrained settings.

Another personalization method, called parameter decoupling, was widely adopted because it could mitigate statistical heterogeneity without substantially increasing communication or computation burdens. A typical design partitioned the model into a “shared base + local personalization module”: e.g., the base + personalization architecture in \cite{arivazhagan2019federated} used shallow layers as a shared base to learn general feature extractors while keeping deeper layers private to capture individual differences; the work in \cite{liang2020think} proposed another layer decoupling that enabled shallow layers to learn local representation alongside global aggregation in deep layers. Parameter decoupling reduced communication by transmitting only a subset of parameters, and the private parameters that were not aggregated enhanced personalization. In federated fine-tuning, FLoRA \cite{Wang_2024_NeurIPS} allowed clients to insert local LoRA adapters of different ranks into the shared foundation model and proposed a stacking aggregation scheme based on matrix operations to achieve correct aggregation under rank heterogeneity, thereby balancing rank heterogeneity and aggregation accuracy; however, FLoRA required uploading adapters for all layers each round, which still incurred high per-round communication overhead in wireless or bandwidth-limited environments. PF2LoRA \cite{Hao_2025_PF2LoRA} trained two LoRA adapter branches, one for global sharing and one for local personalization, in two stages, and employed bilevel optimization to automatically select the rank of local adapters, improving adaptability to data heterogeneity; yet the two-stage design significantly increased per-round computation overhead. Besides, because the global shared adapter structure was consistent across clients, it remained difficult to fully adapt to different wireless conditions for different devices.

Overall, parameter decoupling clearly benefited personalization performance, but most studies viewed it from the perspective of accuracy or personalization and rarely explored its structural value for reducing training costs and adapting to heterogeneous devices.

\subsubsection{Summary}
To sum up, although many compression techniques reduce per-round communication or computation costs, existing low-cost federated learning solutions generally overlook the combined impact of data heterogeneity and device heterogeneity on convergence speed and total training cost; meanwhile, most personalization studies, while improving performance, often introduce extra resource consumption and fail to effectively address communication and computation bottlenecks. More importantly, current research seldom approaches the problem from a system perspective by treating parameter decoupling as a tool to jointly address both data and device heterogeneity with low cost, thereby limiting the overall applicability of federated learning in edge environments.

The rest of this paper is organized as follows. Section II introduces our system model. Section III provides the convergence analysis of the proposed algorithm and formulates the problem accordingly. Section IV discusses several possible future directions. Section V provides the simulation results of our experiments. Section VI concludes the paper.

\section{System Model}

\begin{table}[ht]
\caption{Main Notations}
\centering
\begin{tabular}{c|c}
    \toprule
    Notation & Definition \\
    \midrule
    $\mathcal{N}$ & set of clients \\
    $n, N$ & client index, total number of clients \\
    $\mathcal{D}_n,D_n$ & local dataset, size of local dataset \\
    $\mathcal{X}_n,\mathcal{Y}_n$ & training samples and labels \\
    $\mathbf{x}_{n,i},y_{n,i}$ & single sample and label\\
    $\mathbf{w}_{n,t}$ & client model\\
    $\mathbf{w}_{t}^{\texttt{B}},\mathbf{w}_{n,t}^{\texttt{P}}$ & base and pers. layers\tablefootnote{The ``pers. layers" is the abbreviation for ``personalization layers".} model \\
    $d,d^{\texttt{B}},d^{\texttt{P}}$ & parameter dimensions \\
    $F_n\left( \cdot \right),f\left( \cdot \right)$ & empirical risk, loss function \\
    $\gamma_n$ & ratio of client's training data \\
    $\mathbf{g}_{n,t}$ & stochastic gradients \\
    $\mathbf{g}_{n,t}^{\texttt{B}},\mathbf{g}_{n,t}^{\texttt{P}}$ &stochastic gradients of base and pers. layers  \\
    $\eta$ & learning rate \\
    $\widehat{\mathbf{w}}_{n,t}$ & pruned client model \\
    $\widehat{\mathbf{w}}_{n,t}^{\texttt{P}}$ & pruned pers. layers \\
    $\mathbf{k}_{n,t},\mathbf{r}_{n,t}$ & sparsification strategy, pruning strategy \\
    $k_{n,t},r_{n,t}$ & sparsification rate, pruning rate \\
    $\widehat{\mathbf{g}}_{n,t}$ & pruned stochastic gradients \\
    $\widehat{\mathbf{g}}_{n,t}^{\texttt{B}},\widehat{\mathbf{g}}_{n,t}^{\texttt{P}}$ & pruned stochastic gradients of base and pers. layers \\
    $\widetilde{\mathbf{g}}_{n,t}^{\texttt{B}}$ & gradients of base layer after pruning and sparsification \\
    $T$ & number of training rounds \\
    $p_{n,t}$ & client's transmit power \\
    $h_{n,t}$ & client's channel gain \\
    $l_{n,t}$ & wireless bandwidth allocation \\
    $W$ & total bandwidth of the wireless channel \\
    $N_0$ & power spectral density of AWGN \\
    $FPP$ & floating-point precision \\
    $R_{n,t}$ & client's uplink transmission rate \\
    $S_{n,t}$ & number of bits transmitted by client \\
    $\omega_{n,t}$ & client’s CPU frequency \\
    $\zeta_n$ & energy consumption coefficient of the device's CPU \\
    $C$ & CPU cycles required to process a single sample \\
    $b_n$ & batch size \\
    $L_1,L_2$ & Lipschitz constants \\
    \bottomrule
\end{tabular}
\label{tab:table1}
\end{table}

We consider a MEC scenario where a central server is deployed at a base station (BS), and multiple mobile devices (MDs), such as smartphones, vehicles, and sensors, are regarded as clients. The server coordinates $N$ clients to collaboratively train a machine learning model, such as a convolutional neural network (CNN). Every client $n\in \mathcal{N} := \{ 1,2,\ldots,N\}$ has a local dataset $\mathcal{D}_n = \{\mathcal{X}_n,\mathcal{Y}_n\}$ with $D_{n}$ samples, where $\mathcal{X}_n = \{\mathbf{x}_{n,1},\ldots,\mathbf{x}_{n,i},\ldots,\mathbf{x}_{n,D_n}\}$ represents the training samples, and $\mathcal{Y}_n = \{y_{n,1},\ldots,y_{n,i},\ldots,y_{n,D_n}\}$ represents the sample labels. These datasets are non-IID. During FL training, all clients exchange model parameters with the server via wireless channels. Clients' local computing capabilities and communication resources are typically limited and differ across various devices and training rounds. 

In the proposed federated learning framework, each client first prunes its personalized layers to reduce local computation. The client then calculates gradients of the whole model and updates its personalized parameters. Concurrently, the gradients corresponding to the base layers are sparsified and uploaded to the server. The server aggregates the received gradients to update the global base-layer parameters and subsequently broadcasts the updated shared parameters to all clients. Upon receipt, each client replaces its local base-layer parameters with the broadcasted updates and proceeds to the next training round. This workflow is depicted in Fig.~\ref{fig}.

\begin{figure*}[ht]
\centering
\includegraphics[width=1.0\textwidth]{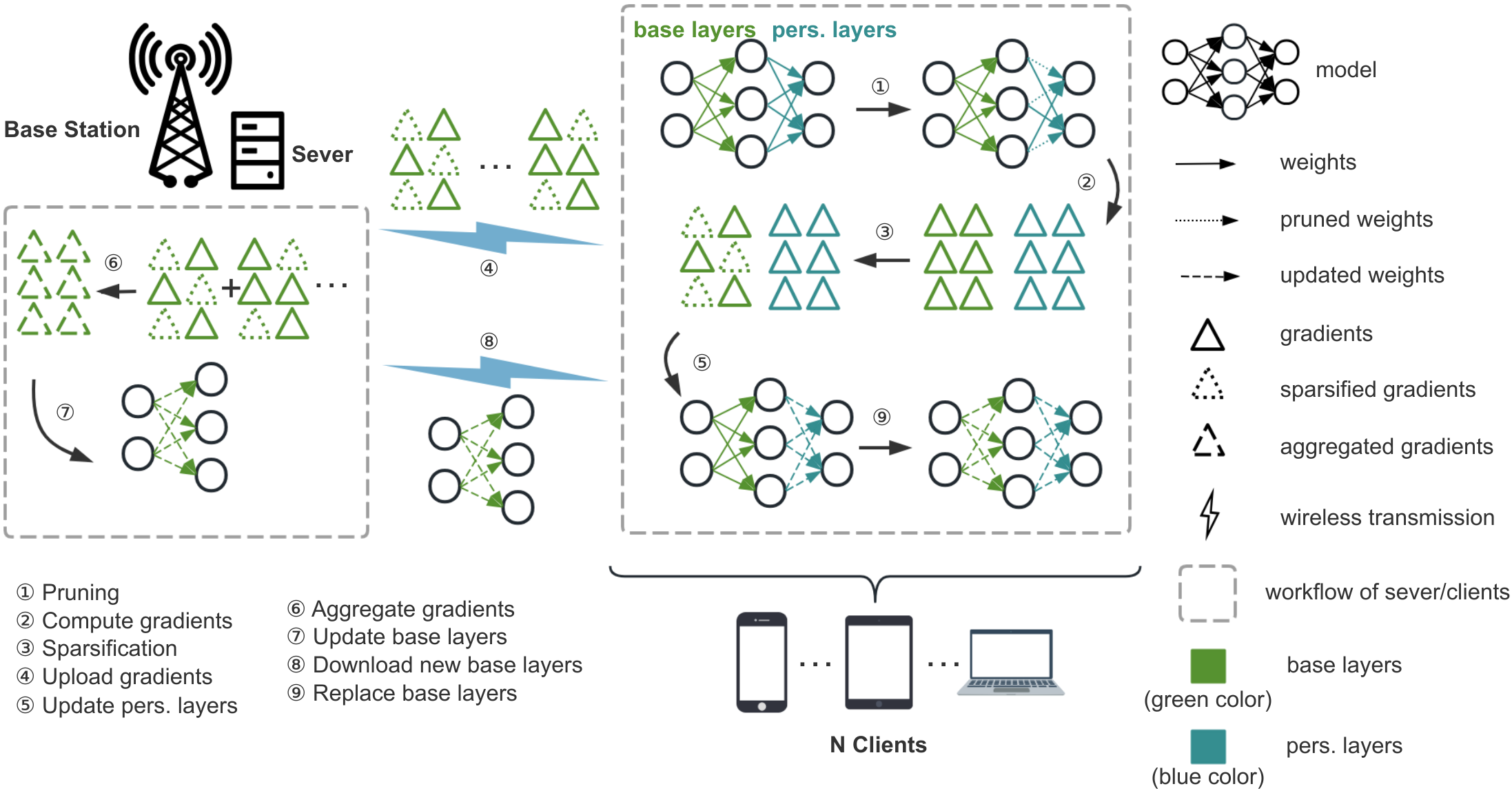}
\caption{The workflow of our proposed approach.}
\label{fig}
\end{figure*}

\subsection{Personalized Federated Learning Model}
Inspired by the concept of parameter decoupling, we divide the client model into base layers and personalization layers to address the statistical heterogeneity in client datasets and achieve model personalization. Although this paper focuses on the lightweight strategy in parameter decoupling-based FL training approach \cite{arivazhagan2019federated}, it is also applicable to similar PFL algorithms that are divided into global and local parts. 

In parameter decoupling-based FL training, with the server's assistance, client $n$ trains a personalized model consisting of base and personalization layers. Without loss of generality, the model weights are represented as a vector $\mathbf{w}_n \in \mathbb{R}^d$, as the work in \cite{9488839,8889996,9691274}. Each client trains its $\mathbf{w}_n$ locally. After parameter decoupling, the shallow layers of the model weights serve as the base layers, while the deep layers serve as the personalization layers. To represent the two parts, $\mathbf{w}_n$ is expanded as $\mathbf{w}_n = \left\lbrack \mathbf{w}^{\texttt{B}}; \mathbf{w}_{n}^{\texttt{P}} \right\rbrack^\texttt{T}$, where $\mathbf{w}^{\texttt{B}} \in \mathbb{R}^{d^{\texttt{B}}}$ denotes the base layers, which are shared by all clients, while $\mathbf{w}_{n}^{\texttt{P}} \in \mathbb{R}^{d^{\texttt{P}}}$ represents the personalization layers, which are retained locally without aggregation or sharing, with $d = d^{\texttt{B}} + d^{\texttt{P}}$. The training objective is formulated as
$\min_{\mathbf{w}_{1},\ldots,\mathbf{w}_N}{\sum_{n = 1}^{N}\gamma_n{F_n\left( \mathbf{w}_n\right)}},
$
where $\gamma_n = D_n/\sum_{n = 1}^{N}D_n$ represents the ratio of training samples held by client $n$ relative to the total number of samples of all clients, and $F_n\left( \mathbf{w}_n \right)$ denotes the empirical risk for client $n$ as
\begin{equation}
F_n\left( \mathbf{w}_n\right) := \frac{1}{D_n}\sum_{i \in \mathcal{D}_n}{f\left( \mathbf{w}_n,\mathbf{x}_{n,i},y_{n,i} \right)},
\end{equation}
where $f\left( \cdot \right)$ is the loss function, and $i$ is the sample index.

In the $t$-th training round, the model is trained using stochastic gradient descent (SGD) with a subset of the samples. Let $\mathcal{D}_{n,t}$ represent a mini-batch subset randomly sampled from the local dataset $\mathcal{D}_n$ during the $t$-th training round, and the batch size is $b_n$. The stochastic gradients for client $n$ at the $t$-th training round are then expressed as
\begin{equation}
 g_n\left( \mathbf{w}_{n,t}\right) = \frac{1}{b_{n}}\sum_{i \in \mathcal{D}_{n,t}}{\nabla f\left( \mathbf{w}_{n,t},\mathbf{x}_{n,i},y_{n,i} \right)}.
\end{equation}
with $\mathbf{w}_{n,t}=\left\lbrack \mathbf{w}_t^{\texttt{B}};\mathbf{w}_{n,t}^{\texttt{P}}\right\rbrack^\texttt{T}$. To simplify the notation, for the stochastic gradients of client $n$, we define ${g_n\left({\mathbf{w}_{n,t}}\right):=\mathbf{g}_{n,t}} = \left\lbrack\mathbf{g}_{n,t}^{\texttt{B}};\mathbf{g}_{n,t}^{\texttt{P}}\right\rbrack^\texttt{T}$. Here, $\mathbf{g}_{n,t} \in \mathbb{R}^{d}$ represents the stochastic gradients as a vector, where $\mathbf{g}_{n,t}^{\texttt{B}} \in \mathbb{R}^{d^{\texttt{B}}}$ represents the component of $\mathbf{g}_{n,t}$ corresponding to $\mathbf{w}_{t}^{\texttt{B}}$, and
$\mathbf{g}_{n,t}^{\texttt{P}} \in \mathbb{R}^{d^{\texttt{P}}}$ represents the component of $\mathbf{g}_{n,t}$ corresponding to $\mathbf{w}_{n,t}^{\texttt{P}}$. The personalization layer weights are updated locally on each client as
$\mathbf{w}_{n,{t+1}}^{\texttt{P}}={\mathbf{w}_{n,t}^{\texttt{P}}}-\eta{\mathbf{g}_{n,t}^{\texttt{P}}},\label{wpup1}
$
where $\eta$ is the learning rate. Then, Clients upload $\mathbf{g}_{n,t}^{\texttt{B}}$ to the server, and the base layer weights are updated on the server as
$\mathbf{w}_{t+1}^{\texttt{B}}=\mathbf{w}_t^{\texttt{B}}-\eta\sum_{n = 1}^{N}{\gamma_n{\mathbf{g}_{n,t}^{\texttt{B}}}}.\label{wbup1}
$

\subsection{Gradient Sparsification and Model Pruning}
In this paper, the gradients of the base layers are sparsified to reduce the amount of data transmitted in the uplink channel, while the personalization layers are subjected to dynamic unstructured pruning\footnote{Pruning can be classified by granularity into structured and unstructured pruning. Structured pruning targets specific network structures, like kernels
or channels in convolutional layers, while unstructured pruning
focuses on individual weights.} to accelerate local computation on the client side. Since the network structure of the FL model varies depending on the specific training task, and different parameter decoupling methods lead to different network structures in the personalization layers, dynamic unstructured pruning is performed on the personalization layers to reduce computation costs.
Sparsification focuses on the gradients obtained after a single training iteration, which can only reduce communication costs without affecting computation costs. 
Hence, sparsification is unnecessary for the personalization layers, as they do not participate in FL aggregation. The gradients of the base layers need to be transmitted and aggregated at the server, and sparsifying the gradients can reduce communication costs. However, pruning the weights in the base layers could lead to inconsistencies in the base layer weights across clients, potentially affecting the sparsification and aggregation of gradients in unpredictable ways. Therefore, pruning is not performed on the base layers.

We first introduce the model pruning of the personalization layers. Specifically, in the $t$-th training round, client $n$ prunes the weights in $\mathbf{w}_{n,t}^{\texttt{P}}$, denoted as
\begin{equation}
\widehat{\mathbf{w}}_{n,t}^{\texttt{P}}=\mathbf{w}_{n,t}^{\texttt{P}} \odot \mathbf{r}_{n,t},
\end{equation}
where $\mathbf{r}_{n,t} \in {\{ 0,1\}}^{d^{\texttt{P}}}$ is the indicator vector representing the pruning strategy, and $\odot$ denotes the Hadamard product, which is the element-wise multiplication of vectors or matrices of the same dimensions. Let $r_{n,t}=\frac{{\|\mathbf{r}_{n,t}\|}_1}{d^{\texttt{P}}}$ be the pruning rate, where ${\|\cdot\|}_1$ represents the $\ell_1$ norm of the vector. Note that $r_{n,t}$ represents the ratio of weights retained after pruning. Once the pruning rate is determined, $\mathbf{r}_{n,t}$ can be generated either by randomly selecting or by choosing $r_{n,t}d^{\texttt{P}}$ weights with the largest absolute values or importance scores, e.g., the square of the product of weights and gradients \cite{9598845}. Then, the client $n$ computes the empirical risk and the stochastic gradients $\widehat{\mathbf{g}}_{n,t} =\lbrack \widehat{\mathbf{g}}_{n,t}^{\texttt{B}};\widehat{\mathbf{g}}_{n,t}^{\texttt{P}}\rbrack^\texttt{T}$ based on the pruned model $\widehat{\mathbf{w}}_{n,t}=\left\lbrack \mathbf{w}_t^{\texttt{B}};\widehat{\mathbf{w}}_{n,t}^{\texttt{P}}\right\rbrack^\texttt{T}$, where
$\widehat{\mathbf{g}}_{n,t}^{\texttt{B}}$ represents the component of $\widehat{\mathbf{g}}_{n,t}$ corresponding to $\mathbf{w}_t^{\texttt{B}}$, and
$\widehat{\mathbf{g}}_{n,t}^{\texttt{P}}$ represents the component of $\widehat{\mathbf{g}}_{n,t}$ corresponding to $\widehat{\mathbf{w}}_{n,t}^{\texttt{P}}$. 
Then, the personalization layer weights are updated locally on client $n$ as follows
\begin{equation}
\begin{split}
\mathbf{w}_{n,{t+1}}^{\texttt{P}}=\widehat{\mathbf{w}}_{n,t}^{\texttt{P}}-\eta\widehat{\mathbf{g}}_{n,t}^{\texttt{P}}.\label{wpup2}
\end{split}
\end{equation}

Secondly, we introduce the gradient sparsification of the base layers. Client $n$ selects part of elements from $\widehat{\mathbf{g}}_{n,t}^{\texttt{B}}$ according to a sparsification strategy and only sends the selected elements to the server, denoted as
\begin{equation}
\widetilde{\mathbf{g}}_{n,t}^{\texttt{B}}= \widehat{\mathbf{g}}_{n,t}^{\texttt{B}} \odot \mathbf{k}_{n,t},\label{sp}
\end{equation}
where $\mathbf{k}_{n,t} \in {\{ 0,1\}}^{d^{\texttt{B}}}$ is the indicator vector representing the sparsification strategy. Let $k_{n,t}=\frac{{\|\mathbf{k}_{n,t}\|}_1}{d^{\texttt{B}}}$ be the sparsification rate. Note that $k_{n,t}$ represents the ratio of gradients retained after sparsification. Once the sparsification rate is determined, $\mathbf{k}_{n,t}$ can be generated by randomly selecting $k_{n,t}$ gradients or by selecting the top $k_{n,t}$ gradients with the largest absolute values. The server collects the $\widetilde{\mathbf{g}}_{n,t}^{\texttt{B}}$ sent by the clients and aggregates them. The base layer weights are updated on the server as follows
\begin{equation}
\begin{split}
\mathbf{w}_{t+1}^{\texttt{B}}=\mathbf{w}_t^{\texttt{B}}-\eta\sum_{n = 1}^{N}{\gamma_n\widetilde{\mathbf{g}}_{n,t}^{\texttt{B}}}.\label{wbup2}
\end{split}
\end{equation} 
Then, the updated base layers are sent to each client. The clients use the updated weights $\mathbf{w}_{n,t+1}=\left\lbrack \mathbf{w}_{t+1}^{\texttt{B}};\mathbf{w}_{n,{t+1}}^{\texttt{P}}\right\rbrack^\texttt{T}$ as the starting point for the next training round. The process repeats until the algorithm converges.

This process is summarized in Algorithm~\ref{alg:alg1}. 
The computational complexity of Algorithm~\ref{alg:alg1} can be decomposed into client-side and server-side costs.
For client-side complexity, the main computational costs per client per round include 1) mask generation (pruning and sparsification):  $O\big(d^{\texttt{P}}\log d^{\texttt{P}} + d^{\texttt{B}}\log d^{\texttt{B}}\big)$;
  2) applying masks  (pruning and sparsification): $O\big(d^{\texttt{P}} + d^{\texttt{B}}\big)$;
  3) gradient computation on the pruned model:  $O\big(b_n (\,d^{\texttt{B}} + r_{n,t}\,d^{\texttt{P}}\,)\big)$; and
  4) local update of personalized parameters:  $O\big(d^{\texttt{P}}\big)$. The dominant term is the gradient computation, hence for $T$ rounds and $N$ participating clients, the computational complexity at the client-side is
$
O\!\Big(T\sum_{n=1}^N b_n\big(d^{\texttt{B}} + r_{n,t}\,d^{\texttt{P}}\big)\Big).
$
This expression clarifies the dependence of client computation on the batch size $b_n$, size of shared and personalized parts, i.e., $d^{\texttt{B}}$, and $d^{\texttt{P}}$, respectively, and the pruning rate $r_{n,t}$.
For server-side complexity, the server's per-round costs include   1) aggregation of received sparse gradients;  2) updating the base-layer parameters;
  and 3) solving the joint resource-compression optimization, i.e., Algorithm~\ref{alg:alg2}, implemented via the DCA procedure.
Assume client $n$ uploads $k_{n,t}d^{\texttt{B}}$ non-zero entries for the base layer in round $t$. The cost to merge these sparse vectors is $O\big(\sum_{n=1}^N k_{n,t}d^{\texttt{B}}\big)$, and the base-layer parameter update costs $O(d^{\texttt{B}})$. 
For the joint optimization, let $I^{\texttt{DCA}}$ denote the number of iterations of the optimization algorithm and $C^{\texttt{IP}}$ the cost of one iteration. Then the complexity of producing pruning and sparsification rates is $O(I^{\texttt{DCA}} \cdot C^{\texttt{IP}})$. Overall, the server-side dominant cost can be expressed as
$O\!\Big(\sum_{n=1}^N k_{n,t}d^{\texttt{B}} \;+\; I^{\texttt{DCA}} \cdot C^{\texttt{IP}})$.

\begin{algorithm}
\caption{Gradient Sparsification and Model Pruning}
\begin{algorithmic}
\STATE \textit{Initialization:}
    \FOR{each client $n\in\mathcal{N}$}
        \STATE Randomly initialize $\mathbf{w}_{n,1}^{\texttt{P}}$
        \STATE Send $D_n$ to the server
    \ENDFOR
\STATE The server receives $D_n$ and calculates $\gamma_n = \frac{D_n}{\sum_{n = 1}^{N}D_n}$
\STATE The server randomly initializes $\mathbf{w}_0^{\texttt{B}},k_{n,0},r_{n,0}$ and sends them to the clients
\STATE
\STATE \textit{Sparsification and Pruning:}
\FOR{each training round $t=0,1,\ldots,T$}
    \FOR{each client $n\in\mathcal{N}$}
        \STATE Receive $\mathbf{w}_t^{\texttt{B}},k_{n,t},r_{n,t}$ from the server
        \STATE Compute $\mathbf{r}_{n,t}$ based on $r_{n,t}$ and $\mathbf{w}_{n,t}^{\texttt{P}}$
        \STATE Prune the pers. layers $\widehat{\mathbf{w}}_{n,t}^{\texttt{P}}=\mathbf{w}_{n,t}^{\texttt{P}}\odot\mathbf{r}_{n,t}$
        \STATE Compute $\widehat{\mathbf{g}}_{n,t}$ base on $\widehat{\mathbf{w}}_{n,t}$ and $\mathcal{D}_{n,t}$
        \STATE Compute $\mathbf{k}_{n,t}$ base on $k_{n,t}$ and $\widehat{\mathbf{g}}_{n,t}^{\texttt{B}}$
        \STATE Sparsify the gradients of base layers $\widetilde{\mathbf{g}}_{n,t}^{\texttt{B}}=\widehat{\mathbf{g}}_{n,t}^{\texttt{B}}\odot\mathbf{k}_{n,t}$
        \STATE Send $\widetilde{\mathbf{g}}_{n,t}^{\texttt{B}}$ to the Server
        \STATE Locally update the pers. layers $\mathbf{w}_{n,{t+1}}^{\texttt{P}}=\widehat{\mathbf{w}}_{n,t}^{\texttt{P}}-\eta\widehat{\mathbf{g}}_{n,t}^{\texttt{P}}$
    \ENDFOR
    \STATE The server receives $\widetilde{\mathbf{g}}_{n,t}^{\texttt{B}}$
    \STATE Update the base layers $\mathbf{w}_{t+1}^{\texttt{B}}=\mathbf{w}_t^{\texttt{B}}-\eta\sum_{n = 1}^{N}{\gamma_n\widetilde{\mathbf{g}}_{n,t}^{\texttt{B}}}$
    \STATE Compute $k_{n,{t+1}},r_{n,{t+1}}$ based on the communication and computation resources of MEC system
    \STATE Send $\mathbf{w}_{t+1}^{\texttt{B}},k_{n,{t+1}},r_{n,{t+1}}$ to the clients
\ENDFOR
\end{algorithmic}
\label{alg:alg1}
\end{algorithm}

\subsection{Wireless Communication Model}
To address the issue of device heterogeneity, it is necessary to optimize the clients' sparsification rates, pruning rates, and resource allocation in each training round of FL training. Note that each training round corresponds to one communication round, as the client and server exchange information only once per iteration. The latency and energy consumption of the FL training are estimated using the following wireless communication model. Clients sparsify the stochastic gradients of base layers and upload the sparse gradients to the server. Assuming frequency-division multiple access (FDMA) is used for uplink transmission, $p_{n,t}$ (Watt) represents the uplink transmission power of client $n$ in $t$-th training round, and $h_{n,t}$ represents the wireless channel gain between client $n$ and the server, and $l_{n,t}$ represents the wireless bandwidth allocation for client $n$, where $l_{n,t} \in [0, 1]$. In the $t$-th training round, the uplink transmission rate (bit/s) of client $n$ can be expressed as follows
\begin{equation}
R_{n,t}=l_{n,t}W\log_2\left(1+\frac{ h_{n,t} p_{n,t} }{N_0l_{n,t}W}\right),
\end{equation}
where $W$ (Hz) is the total bandwidth of the uplink wireless channel, and $N_0$ (W/Hz) is the power spectral density of the additive white Gaussian noise (AWGN).

Through gradient sparsification in Eq.~\eqref{sp}, 
there are $\left(1-k_{n,t}\right)d^{\texttt{B}}$ elements set to zero in $\widetilde{\mathbf{g}}_{n,t}^{\texttt{B}}$, while $k_{n,t} d^{\texttt{B}}$ elements are retained. To improve transmission efficiency, clients transmit only the non-zero elements and their positions rather than the entire vector. For each non-zero element, two parts need to be uploaded. One part is the value of the element, which depends on the floating-point precision ($FPP$), such as $FPP=32$ bits for single precision or $FPP=64$ bits for double precision, and 1 bit is needed to represent the sign of the element. Another part is the position of the element. In the $t$-th training round, the number of bits transmitted by client $n$ in the uplink transmission can be expressed as
\begin{equation}
\begin{split}
S_{n,t} = k_{n,t}d^{\texttt{B}}\left(FPP+1\right)+\log_2\binom{d^{\texttt{B}}}{k_{n,t}d^{\texttt{B}}},
\end{split}
\end{equation}
where $\binom{\cdot}{\cdot}$ represents a binomial coefficient. Referring to \cite{9488839}\footnote{The position of the element can be approximated as
\begin{equation}
\begin{split}
\log_2\binom{d^{\texttt{B}}}{k_{n,t}d^{\texttt{B}}}=\log_2d^{\texttt{B}}!-\log_2\left(k_{n,t}d^{\texttt{B}}\right)!-\log_2\left(d^{\texttt{B}}-k_{n,t}d^{\texttt{B}}\right)!\\\overset{(a)}{\approx}d^{\texttt{B}}\lbrack\log_2\frac{1}{k_{n,t}}+\left(k_{n,t}-1\right)\log_2\left(\frac{1}{k_{n,t}}-1\right)\rbrack\overset{(b)}{\approx}k_{n,t}d^{\texttt{B}}\log_2\frac{1}{k_{n,t}} \nonumber
\end{split}
\end{equation}
(a) is by Stirling formula that gives precise estimates for factorials, i.e., $\log_2n!\approx n\log_2n-n\log_2e$, and (b) is due to $k_{n,t}\ll1$.
}, it can be approximately simplified as
\begin{equation}
\begin{split}
S_{n,t}\approx k_{n,t}d^{\texttt{B}}\left(\log_2{\frac{1}{k_{n,t}}+FPP+1}\right).
\end{split}
\end{equation}
After the server updates the base layers, BS broadcasts the updated weights to all clients. Since BS has a high transmission power and uses the entire bandwidth for broadcasting, the downlink communication latency can be neglected compared to the uplink.
Thus, the communication latency (second) for client $n$ in the $t$-th training round is given by
\begin{equation}
\tau_{n,t}^{\texttt{Comm}}=\frac{S_{n,t}}{R_{n,t}}.
\end{equation}
The communication energy consumption (Joule) for client $n$ in the $t$-th training round is given by
\begin{equation}
E_{n,t}^{\texttt{Comm}}=p_{n,t}\tau_{n,t}^{\texttt{Comm}}.
\end{equation}
Since the BS generally has a continuous power supply, its communication energy consumption in the downlink is not considered.

The actual computation latency of client $n$ depends on various factors, such as the neural network model structure, memory bandwidth, cache hit rate, etc. Referring to \cite{9598845,9264742}, we estimate the computation latency and energy consumption. Suppose the number of CPU cycles required to process one sample is $C$ (cycles/sample), the computation capability refers to the CPU frequency of client $n$ in $t$-th training round is $\omega_{n,t}$ (cycles/s or Hz), the energy consumption coefficient of the device's CPU is $\zeta_n$ ($\text{J}\cdot\text{s}^2$), and the batch size is $b_n$ (samples). Then, the computation latency (second) of client $n$ in $t$-th training round is given by
\begin{equation}
\tau_{n,t}^{\texttt{Comp}}=\frac{b_nC\left(d^{\texttt{B}}+r_{n,t}d^{\texttt{P}}\right)}{\omega_{n,t}d}.
\end{equation} 
The computation energy consumption (Joule) of client $n$ in the $t$-th training round is given by
\begin{equation}
E_{n,t}^{\texttt{Comp}}=\zeta_n\omega_{n,t}^3\tau_{n,t}^{\texttt{Comp}}.
\end{equation}

To sum up, the total latency (second) for client $n$ of one training round is given by
\begin{equation}
\tau_{n,t}^{\texttt{all}} =\tau_{n,t}^{\texttt{Comp}}+\tau_{n,t}^{\texttt{Comm}}.
\end{equation}
The total energy consumption (Joule) for client $n$ of all iterations is given by
\begin{equation}
E_n^{\texttt{all}} =\sum_{t=1}^T\left( E_{n,t}^{\texttt{Comp}}+E_{n,t}^{\texttt{Comm}}\right).
\end{equation}
\section{Convergence Analysis and Problem Formulation}
\subsection{Convergence Analysis}
Since the empirical risk (loss function) of neural networks is typically non-convex, the convergence of an algorithm can be estimated using the average $\ell_2$ norm of the gradients of the empirical risk \cite{doi:10.1137/16M1080173,9598845}. To determine the convergence of our algorithm, we make the following assumptions. 
\begin{assumption} {\rm (L-Smooth) \label{ass1} The empirical risk function $F_n\left(\cdot\right)$ for client $n$ is L-smooth, then for every $\mathbf{x},\mathbf{y}\in\mathbb{R}^d$, there exist constants $L_1, L_2 > 0$ such that}
\begin{equation}
\label{19}
\left\|F_n\left( \mathbf{x} \right) - F_n\left( \mathbf{y} \right) \right\|\ \leq \ L_1\left\| \mathbf{x} - \mathbf{y} \right\|,
\end{equation}
\begin{equation}
\label{20}
\left\| \nabla F_n\left( \mathbf{x} \right) - \nabla F_n\left( \mathbf{y} \right) \right\|\ \leq \ L_2\left\| \mathbf{x} - \mathbf{y} \right\|,
\end{equation}
where $\left\|\cdot\right\|$ represents the $\ell_2$-norm, and~\eqref{20} implies that
\begin{equation}
\label{21}
F_n\left( \mathbf{x} \right) \leq F_n\left( \mathbf{y} \right) + \left\langle \nabla F_n\left( \mathbf{y} \right),\mathbf{x - y} \right\rangle + \frac{L_2}{2}\left\| \mathbf{x - y} \right\|^{2}.
\end{equation}
\end{assumption}
\begin{assumption} \label{ass2}
{\rm (Unbiased Gradients and Bounded Variances) 
The stochastic gradients are unbiased, and the variances of the stochastic gradients are bounded by $\sigma^2$
, then for every $n,t$, the following conditions hold}
\begin{equation}
\label{222}
\begin{split}
\nabla F_n\left( \mathbf{w}_{n,t}\right)=\mathbb{E}\lbrack\mathbf{g}_{n,t}\rbrack,
\end{split}
\end{equation}
\begin{equation}
\label{22}
\begin{split}
\mathbb{E}\left\|\nabla F_n\left( \mathbf{w}_{n,t}\right) - \mathbf{g}_{n,t}\right\|^2 \leq \sigma^{2},
\end{split}
\end{equation}
where $\nabla F_n\left(\mathbf{w}_{n,t}\right)$ represents the gradients using the entire dataset $\mathcal{D}_{n}$ of the client. $\mathbb{E}\lbrack\cdot\rbrack$ denotes the expectation with respect to the SGD sampling over $\mathcal{D}_{n,t}$.
\end{assumption}
\begin{assumption} {\rm (Bounded Gradients and Weights) \label{ass3}
Both the stochastic gradients and the model weights are bounded, then for every $n,t$, there exist $G>0$ and $M>0$ such that}
\begin{equation}
\label{23}
\begin{split}
\mathbb{E}\left\|\mathbf{g}_{n,t}\right\|^2 \leq G^{2},
\end{split}
\end{equation}
\begin{equation}
\label{24}
\begin{split}
\left\| \mathbf{w}_{n,t} \right\|^2 \leq M^{2}.
\end{split}
\end{equation}
\end{assumption}
Besides, the following lemma is required.
\begin{lemma}\label{lem1}
{\rm For $\mathbf{x} \in \mathbb{R}^d$, $\mathbf{k} \in {\{ 0,1\}}^d$, where $\mathbf{k}$ contains $k$ elements equal to $1$, with $1\leq k\leq d$,  according to the work 
\cite{NEURIPS2018_b440509a},
it follows that 
}
\begin{equation}
\label{36}
\begin{split}
\mathbb{E}_{\mathbf{k}}\|\mathbf{x}-\mathbf{x}\odot\mathbf{k}\|^2 = \left(1-\frac{k}{d}\right)\|\mathbf{x}\|^2.
\end{split}
\end{equation}

\begin{proof}
{\rm $\mathbf{k}$ has a total of $\binom{k}{d}$ possible values, which form the set $\mathcal{K}$, Let $k_i$ be an element in $\mathbf{k}$, and it follows that}
\begin{equation*}
\begin{split}
&\mathbb{E}_{\mathbf{k}}\|\mathbf{x}-\mathbf{x}\odot\mathbf{k}\|^2 = \frac{1}{\left|\mathcal{K}\right|}\sum_{\mathbf{k}\in\mathcal{K}}\sum_{i=1}^d x_i^2\left(1-k_i\right)
\\&=\sum_{i=1}^d x_i^2\sum_{\mathbf{k}\in\mathcal{K}}\frac{1-k_i}{\left|\mathcal{K}\right|}=\left(1-\frac{k}{d}\right)\|\mathbf{x}\|^2.
\end{split}
\end{equation*}
\end{proof}

{\rm Using the approach analogous to the proof of Eq.~\eqref{36}, Eq.~\eqref{37} can be established. }

\begin{equation}
\label{37}
\begin{split}
\mathbb{E}_{\mathbf{k}}\|\mathbf{x}\odot\mathbf{k}\|^2= \frac{k}{d}\|\mathbf{x}\|^2.
\end{split}
\end{equation}

\begin{proof}
{\rm $\mathbf{k}$ has a total of $\binom{k}{d}$ possible values, which form the set $\mathcal{K}$, Let $k_i$ be an element in $\mathbf{k}$, and it follows that}
\begin{equation*}
\begin{split}
&\mathbb{E}_{\mathbf{k}}\|\mathbf{x}\odot\mathbf{k}\|^2 = \frac{1}{\left|\mathcal{K}\right|}\sum_{\mathbf{k}\in\mathcal{K}}\sum_{i=1}^d x_i^2 k_i\\&=\sum_{i=1}^d x_i^2\sum_{\mathbf{k}\in\mathcal{K}}\frac{k_i}{\left|\mathcal{K}\right|}=\frac{k}{d}\|\mathbf{x}\|^2.
\end{split}
\end{equation*}

\end{proof}

\end{lemma}


For given sparsification and pruning rates $k_{n,t}$ and $r_{n,t}$, there are different strategies to determine the sparsification and pruning indicators $\mathbf{k}_{n,t}$ and $\mathbf{r}_{n,t}$ (e.g., random, magnitude-based, importance-based selection). Without loss of generality, it is assumed that $\mathbf{k}_{n,t}$ and $\mathbf{r}_{n,t}$ are random variables, following a uniform distribution.
Since SGD performs random sampling on $\mathcal{D}_{n,t}$, the pruned and updated weights $\widehat{\mathbf{w}}_{n,t}$ and $\mathbf{w}_{n,{t+1}}$ are also random variables for a given $\mathbf{w}_{n,t}$ during the $t$-th training round of client $n$. Assume that $\mathcal{D}_{n,t}$, $\mathbf{k}_{n,t}$ and $\mathbf{r}_{n,t}$ are independent.
The notation $\mathbb{E}_{\mathcal{D}_{n,t}}\lbrack\cdot\rbrack$ denotes the expectation with respect to the random sampling of $\mathcal{D}_{n,t}$ using SGD for given $\mathbf{w}_{n,t}$. The notations  $\mathbb{E}_{\mathbf{k}_{n,t}}\lbrack\cdot\rbrack$ and $\mathbb{E}_{\mathbf{r}_{n,t}}\lbrack\cdot\rbrack$ denote the expectation with respect to $\mathbf{k}_{n,t}$ and $\mathbf{r}_{n,t}$, respectively. For ease of notation, $\mathbb{E}_{\mathcal{D}_{n,t},\mathbf{k}_{n,t},\mathbf{r}_{n,t}}\lbrack\cdot\rbrack$ and $\nabla F_n\left( \widehat{\mathbf{w}}_{n,t}\right) $ are represented by $\mathbb{E}\lbrack\cdot\rbrack$ and $\nabla\widehat{\mathbf{F}}_{n,t}=\left\lbrack \nabla \widehat{\mathbf{F}}_{n,t}^{\texttt{B}};\nabla \widehat{\mathbf{F}}_{n,t}^{\texttt{P}} \right\rbrack^\texttt{T}$, respectively. Based on the above assumptions and Lemma~\ref{lem1}, we can derive the following propositions, which are the foundations for the final convergence analysis.

Before commencing the convergence analysis, we first provide a high-level overview of the proof process. Specifically, Proposition~\ref{proposition3} derives a key one-round empirical risk descent inequality, in which the perturbation terms introduced by model pruning and gradient sparsification (denoted as $\kappa_1$ and $\kappa_2$, respectively) are explicitly separated. This inequality serves as the cornerstone of our entire analysis. Subsequently, in Propositions~\ref{proposition1}  and~\ref{proposition2}, we conduct a detailed investigation of these two perturbation terms and establish explicit upper bounds determined by the pruning rate and sparsification rate. Finally, in Theorem~\ref{lem1} , we substitute these bounds back into the initial inequality and, by summing and averaging across all iterations, obtain the global convergence guarantee of the proposed algorithm.

\begin{proposition}
\label{proposition3}
{\rm After the parameters are decoupled, the difference of the empirical risks between the two rounds is upper-bounded. Specifically, the following inequality holds} 
\begin{equation}
\label{pro3}
\begin{split}
&\mathbb{E}\lbrack F_n\left( \mathbf{w}_{n,{t+1}} \right)\rbrack- F_n\left( \mathbf{w}_{n,t} \right)\\&\leq 
\kappa_1\left(\mathbf{r}_{n,t}\right)
+
\kappa_2\left(\mathbf{k}_{n,t}\right)
-\frac{\eta}{2}\mathbb{E}_{\mathbf{r}_{n,t}}\left\|\nabla\widehat{\mathbf{F}}_{n,t} \right\|^2,
\end{split}
\end{equation}
{\rm where}
\begin{equation}
\label{pro1}
\begin{split}
\kappa_1\left(\mathbf{r}_{n,t}\right) &=L_1\mathbb{E}_{\mathbf{r}_{n,t}}\left\|\mathbf{w}_{n,t}-\widehat{\mathbf{w}}_{n,t}\right\|\\&+\frac{L_2\eta^2-\eta}{2}\mathbb{E}_{\mathcal{D}_{n,t},\mathbf{r}_{n,t}}\left\|\widehat{\mathbf{g}}_{n,t}^{\texttt{P}}\right\|^2
\\&+\eta\mathbb{E}_{\mathcal{D}_{n,t},\mathbf{r}_{n,t}}\left\| \nabla\widehat{\mathbf{F}}_{n,t}^{\texttt{B}}-\sum_{i = 1}^{N}{\gamma_i\widehat{\mathbf{g}}_{i,t}^{\texttt{B}}}\right\|^2
\\&+\frac{\eta}{2}\mathbb{E}_{\mathcal{D}_{n,t},\mathbf{r}_{n,t}}\left\|\nabla\widehat{\mathbf{F}}_{n,t}^{\texttt{P}}-\widehat{\mathbf{g}}_{n,t}^{\texttt{P}}\right\|^2,
\end{split}
\end{equation}
{\rm and }
\begin{equation}
\label{pro2}
\begin{split}
\kappa_2\left(\mathbf{k}_{n,t}\right)&=
\frac{L_2\eta^2-\eta}{2}\mathbb{E}\left\|\sum_{i = 1}^{N}{\gamma_i\widetilde{\mathbf{g}}_{i,t}^{\texttt{B}}}\right\|^2\\&+\eta\mathbb{E}\left\|\sum_{i = 1}^{N}\gamma_i\left(\widehat{\mathbf{g}}_{i,t}^{\texttt{B}}- \widetilde{\mathbf{g}}_{i,t}^{\texttt{B}}\right)\right\|^2.
\end{split}
\end{equation}
\end{proposition}

\begin{proof}
According to Assumption~\ref{ass1}, taking the expectation on both sides of Eq.~\eqref{21} and setting 
$\mathbf{x} = \mathbf{w}_{n,{t+1}}$ and $\mathbf{y} = \widehat{\mathbf{w}}_{n,t}$ give
\begin{equation*}
\begin{split}
&\mathbb{E}\lbrack F_n\left(\mathbf{w}_{n,{t+1}}\right)\rbrack \leq \mathbb{E}\lbrack F_n\left( \widehat{\mathbf{w}}_{n,t}\right)\rbrack \\&+ \mathbb{E}\left\langle \nabla\widehat{\mathbf{F}}_{n,t},\mathbf{w}_{n,{t+1}} - \widehat{\mathbf
{w}}_{n,t} \right\rangle\\&+ \frac{L_2}{2}\mathbb{E}\left\| \mathbf{w}_{n,{t+1}} - \widehat{\mathbf{w}}_{n,t} \right\|^{2}.
\end{split}
\end{equation*}
From Eq.~\eqref{19} in Assumption~\ref{ass1}, $\mathbb{E}\lbrack F_n\left( \widehat{\mathbf{w}}_{n,t}\right)\rbrack\leq F_n\left( \mathbf{w}_{n,t} \right)+L_1\mathbb{E}_{\mathbf{r}_{n,t}}\left\|\mathbf{w}_{n,t}-\widehat{\mathbf{w}}_{n,t}\right\|$. Therefore, 
\begin{equation*}
\begin{split}
&\mathbb{E}\lbrack F_n\left( \mathbf{w}_{n,{t+1}} \right)\rbrack-F_n\left( \mathbf{w}_{n,t} \right)\leq L_1\mathbb{E}_{\mathbf{r}_{n,t}}\left\|\mathbf{w}_{n,t}-\widehat{\mathbf{w}}_{n,t}\right\|
\\&+ \mathbb{E}\left\langle \nabla\widehat{\mathbf{F}}_{n,t},\mathbf{w}_{n,{t+1}} - \widehat{\mathbf
{w}}_{n,t} \right\rangle
\\&+ \frac{L_2}{2}\mathbb{E}\left\| \mathbf{w}_{n,{t+1}} - \widehat{\mathbf{w}}_{n,t} \right\|^{2}.
\end{split}
\end{equation*}
From the Eq.~\eqref{wpup2} and Eq.~\eqref{wbup2}, it follows that
\begin{equation*}
\begin{split}
&\mathbb{E}\lbrack F_n\left( \mathbf{w}_{n,{t+1}} \right)\rbrack-F_n\left( \mathbf{w}_{n,t} \right) \leq L_1\mathbb{E}_{\mathbf{r}_{n,t}}\left\|\mathbf{w}_{n,t}-\widehat{\mathbf{w}}_{n,t}\right\| 
\\&-\eta \mathbb{E}\left\langle \nabla\widehat{\mathbf{F}}_{n,t},\lbrack \sum_{i = 1}^{N}{\gamma_i\widetilde{\mathbf{g}}_{i,t}^{\texttt{B}}};\widehat{\mathbf{g}}_{n,t}^{\texttt{P}} \rbrack^\texttt{T} \right\rangle
\\&+\frac{L_2\eta^2}{2}\mathbb{E}\left\| \lbrack \sum_{i = 1}^{N}{\gamma_i\widetilde{\mathbf{g}}_{i,t}^{\texttt{B}}};\widehat{\mathbf{g}}_{n,t}^{\texttt{P}} \rbrack^\texttt{T} \right\|^{2}.
\end{split}
\end{equation*}
From $\left\langle \mathbf{a},\mathbf{b}\right\rangle=\frac{1}{2}\left\|\mathbf{a} \right\|^2+\frac{1}{2}\left\|\mathbf{b}\right\|^2-\frac{1}{2}\left\|\mathbf{a}-\mathbf{b} \right\|^2$, it follows that
\begin{equation*}
\begin{split}
&\mathbb{E}\lbrack F_n\left( \mathbf{w}_{n,{t+1}} \right)\rbrack-F_n\left( \mathbf{w}_{n,t} \right) \leq L_1\mathbb{E}_{\mathbf{r}_{n,t}}\left\|\mathbf{w}_{n,t}-\widehat{\mathbf{w}}_{n,t}\right\|\\& -\frac{\eta}{2}\mathbb{E}_{\mathbf{r}_{n,t}}\left\|\nabla\widehat{\mathbf{F}}_{n,t} \right\|^2+\frac{L_2\eta^2-\eta}{2}\mathbb{E}\left\|\lbrack \sum_{i = 1}^{N}{\gamma_i\widetilde{\mathbf{g}}_{i,t}^{\texttt{B}}};\widehat{\mathbf{g}}_{n,t}^{\texttt{P}} \rbrack^\texttt{T}\right\|^2\\&+\frac{\eta}{2}\mathbb{E}\left\|\lbrack \nabla\widehat{\mathbf{F}}_{n,t}^{\texttt{B}}-\sum_{i = 1}^{N}{\gamma_i\widetilde{\mathbf{g}}_{i,t}^{\texttt{B}}};\nabla\widehat{\mathbf{F}}_{n,t}^{\texttt{P}}-\widehat{\mathbf{g}}_{n,t}^{\texttt{P}} \rbrack^\texttt{T}\right\|^2.
\end{split}
\end{equation*}
From $\left\|\lbrack\mathbf{a};\mathbf{b}\rbrack^\texttt{T}\right\|^2=\left\|\mathbf{a}\right\|^2+\left\|\mathbf{b}\right\|^2$, it follows that
\begin{equation*}
\begin{split}
&\mathbb{E}\lbrack F_n\left( \mathbf{w}_{n,{t+1}} \right)\rbrack-F_n\left( \mathbf{w}_{n,t} \right) \leq L_1\mathbb{E}_{\mathbf{r}_{n,t}}\left\|\mathbf{w}_{n,t}-\widehat{\mathbf{w}}_{n,t}\right\| 
\\&-\frac{\eta}{2}\mathbb{E}_{\mathbf{r}_{n,t}}\left\|\nabla\widehat{\mathbf{F}}_{n,t} \right\|^2
+\frac{L_2\eta^2-\eta}{2}\mathbb{E}\left\|\sum_{i = 1}^{N}{\gamma_i\widetilde{\mathbf{g}}_{i,t}^{\texttt{B}}}\right\|^2
\\&+\frac{L_2\eta^2-\eta}{2}\mathbb{E}_{\mathcal{D}_{n,t},\mathbf{r}_{n,t}}\left\|\widehat{\mathbf{g}}_{n,t}^{\texttt{P}}\right\|^2+\frac{\eta}{2}\mathbb{E}\left\| \nabla\widehat{\mathbf{F}}_{n,t}^{\texttt{B}}-\sum_{i = 1}^{N}{\gamma_i\widetilde{\mathbf{g}}_{i,t}^{\texttt{B}}}\right\|^2\\&+\frac{\eta}{2}\mathbb{E}_{\mathcal{D}_{n,t},\mathbf{r}_{n,t}}\left\|\nabla\widehat{\mathbf{F}}_{n,t}^{\texttt{P}}-\widehat{\mathbf{g}}_{n,t}^{\texttt{P}}\right\|^2.
\end{split}
\end{equation*}
By Jensen’s inequality, we have $\|\mathbf{a}+\mathbf{b}\|^2\leq2\|\mathbf{a}\|^2+2\|\mathbf{b}\|^2$, then it follows that
\begin{equation*}
\begin{split}
&\mathbb{E}\lbrack F_n\left( \mathbf{w}_{n,{t+1}} \right)\rbrack-F_n\left( \mathbf{w}_{n,t} \right) \leq L_1\mathbb{E}_{\mathbf{r}_{n,t}}\left\|\mathbf{w}_{n,t}-\widehat{\mathbf{w}}_{n,t}\right\|
\\&-\frac{\eta}{2}\mathbb{E}_{\mathbf{r}_{n,t}}\left\|\nabla\widehat{\mathbf{F}}_{n,t} \right\|^2+\frac{L_2\eta^2-\eta}{2}\mathbb{E}\left\|\sum_{i = 1}^{N}{\gamma_i\widetilde{\mathbf{g}}_{i,t}^{\texttt{B}}}\right\|^2
\\&+\frac{L_2\eta^2-\eta}{2}\mathbb{E}_{\mathcal{D}_{n,t},\mathbf{r}_{n,t}}\left\|\widehat{\mathbf{g}}_{n,t}^{\texttt{P}}\right\|^2+\eta\mathbb{E}\left\|\sum_{i =1}^{N}\gamma_i\left(\widehat{\mathbf{g}}_{i,t}^{\texttt{B}}-\widetilde{\mathbf{g}}_{i,t}^{\texttt{B}}\right)\right\|^2\\&+\eta\mathbb{E}_{\mathcal{D}_{n,t},\mathbf{r}_{n,t}}\left\| \nabla\widehat{\mathbf{F}}_{n,t}^{\texttt{B}}-\sum_{i = 1}^{N}{\gamma_i\widehat{\mathbf{g}}_{i,t}^{\texttt{B}}}\right\|^2\\&+\frac{\eta}{2}\mathbb{E}_{\mathcal{D}_{n,t},\mathbf{r}_{n,t}}\left\|\nabla\widehat{\mathbf{F}}_{n,t}^{\texttt{P}}-\widehat{\mathbf{g}}_{n,t}^{\texttt{P}}\right\|^2.
\end{split}
\end{equation*}
Denote functions related to the effect of pruning and sparsification 
\begin{equation*}
\begin{split}
&L_1\mathbb{E}_{\mathbf{r}_{n,t}}\left\|\mathbf{w}_{n,t}-\widehat{\mathbf{w}}_{n,t}\right\|+\frac{L_2\eta^2-\eta}{2}\mathbb{E}_{\mathcal{D}_{n,t},\mathbf{r}_{n,t}}\left\|\widehat{\mathbf{g}}_{n,t}^{\texttt{P}}\right\|^2
\\&+\eta\mathbb{E}_{\mathcal{D}_{n,t},\mathbf{r}_{n,t}}\left\| \nabla\widehat{\mathbf{F}}_{n,t}^{\texttt{B}}-\sum_{i = 1}^{N}{\gamma_i\widehat{\mathbf{g}}_{i,t}^{\texttt{B}}}\right\|^2
\\&+\frac{\eta}{2}\mathbb{E}_{\mathcal{D}_{n,t},\mathbf{r}_{n,t}}\left\|\nabla\widehat{\mathbf{F}}_{n,t}^{\texttt{P}}-\widehat{\mathbf{g}}_{n,t}^{\texttt{P}}\right\|^2
\end{split}
\end{equation*}
and 
\begin{equation*}
\begin{split}
\frac{L_2\eta^2-\eta}{2}\mathbb{E}\left\|\sum_{i = 1}^{N}{\gamma_i\widetilde{\mathbf{g}}_{i,t}^{\texttt{B}}}\right\|^2+\eta\mathbb{E}\left\|\sum_{i = 1}^{N}\gamma_i\left(\widehat{\mathbf{g}}_{i,t}^{\texttt{B}}- \widetilde{\mathbf{g}}_{i,t}^{\texttt{B}}\right)\right\|^2
\end{split}
\end{equation*}
as $\kappa_1\left(\mathbf{r}_{n,t}\right)$ and $\kappa_2\left(\mathbf{k}_{n,t}\right)$, respectively. The difference of the empirical risks between the two rounds is upper-bounded as follows

\begin{equation*}
\begin{split}
&\mathbb{E}\lbrack F_n\left( \mathbf{w}_{n,{t+1}} \right)\rbrack- F_n\left( \mathbf{w}_{n,t} \right)\\&\leq \kappa_1\left(\mathbf{r}_{n,t}\right)+\kappa_2\left(\mathbf{k}_{n,t}\right)-\frac{\eta}{2}\mathbb{E}_{\mathbf{r}_{n,t}}\left\|\nabla\widehat{\mathbf{F}}_{n,t} \right\|^2
.
\end{split}
\end{equation*}
\end{proof}
\begin{remark}
{\rm The $\ell_2$ norm of the gradients depends on the difference in the model weights between two consecutive iterations. Due to the introduction of parameter decoupling, sparsification, and pruning, the weights update are influenced by these factors. Quantifying their effects requires us to overcome the challenges of analyzing the mapping of parameter decoupling on the weights update and the combined effects of these three factors.}
\end{remark}
\begin{remark}
{\rm Although \cite{NEURIPS2018_d54e99a6} points out that the exact computation of the Lipschitz constant in deep learning architectures is challenging, we can still approximate the bounds for $L_1$ and $L_2$ by referring to \cite{NEURIPS2018_d54e99a6, khromov2023some}. The other constants $G, M,\sigma$ can be computed using local training data on MDs and then obtained by taking the maximum expected value across all devices.}
\end{remark}

Proposition~\ref{proposition3} provides the foundational framework for our convergence analysis; however, the term $\kappa_1\left(\mathbf{r}_{n,t}\right)$, which represents the effect of pruning, remains a complex expression that depends on multiple sources of randomness. To advance the analysis, the next critical step is to establish a simplified and interpretable upper bound for this term, thereby quantifying the specific impact of pruning on convergence. Accordingly, in Proposition~\ref{proposition1}, we focus on analyzing $\kappa_1\left(\mathbf{r}_{n,t}\right)$ and derive an upper bound that depends solely on the pruning rate $r_{n,t}$ and fundamental system parameters.

\begin{proposition}
\label{proposition1}
{\rm  For the given pruning rate $r_{n,t}$, the effect of the pruning of the personalization layers on the convergence $\kappa_1\left(\mathbf{r}_{n,t}\right)$ is upper-bounded by $\Psi_1(r_{n,t})$, with $\Psi_1(r_{n,t})$ given by }

\begin{equation}
\begin{split}
\Psi_1(r_{n,t})=L_1M\sqrt{\left(1-r_{n,t}\right)}+\frac{L_2\eta^2-\eta}{2}G^2+\frac{3\eta\sigma^{2}}{2}.  
\end{split}
\end{equation}
\end{proposition}

\begin{proof}
From Jensen’s inequality, it follows that
\begin{equation*}
\begin{split}
&\kappa_1\left(\mathbf{r}_{n,t}\right)\leq L_1\mathbb{E}_{\mathbf{r}_{n,t}}\left\|\mathbf{w}_{n,t}-\widehat{\mathbf{w}}_{n,t}\right\|\\&+\frac{L_2\eta^2-\eta}{2}\mathbb{E}_{\mathcal{D}_{n,t},\mathbf{r}_{n,t}}\left\|\widehat{\mathbf{g}}_{n,t}^{\texttt{P}}\right\|^2
\\&+\eta\sum_{i = 1}^{N}{\gamma_i\mathbb{E}_{\mathcal{D}_{n,t},\mathbf{r}_{n,t}}\left\| \nabla\widehat{\mathbf{F}}_{n,t}^{\texttt{B}}-\widehat{\mathbf{g}}_{i,t}^{\texttt{B}}\right\|^2}
\\&+\frac{\eta}{2}\mathbb{E}_{\mathcal{D}_{n,t},\mathbf{r}_{n,t}}\left\|\nabla\widehat{\mathbf{F}}_{n,t}^{\texttt{P}}-\widehat{\mathbf{g}}_{n,t}^{\texttt{P}}\right\|^2.
\end{split}
\end{equation*}

From Eq.~\eqref{23} in Assumption~\ref{ass3}, we have 
$\mathbb{E}_{\mathcal{D}_{n,t},\mathbf{r}_{n,t}}\left\|\widehat{\mathbf{g}}_{n,t}^{\texttt{P}}\right\|^2=\mathbb{E}_{\mathbf{r}_{n,t}}\left( \mathbb{E}_{\mathcal{D}_{n,t}}\left\|\widehat{\mathbf{g}}_{n,t}^{\texttt{P}}\right\|^2\right)\leq\mathbb{E}_{\mathbf{r}_{n,t}}G^2=G^2$, Similarly, from  Eq.~\eqref{22} in Assumption~\ref{ass2}, we have 
$\mathbb{E}_{\mathcal{D}_{n,t},\mathbf{r}_{n,t}}\left\| \nabla\widehat{\mathbf{F}}_{n,t}^{\texttt{B}}-\widehat{\mathbf{g}}_{i,t}^{\texttt{B}}\right\|^2\leq\sigma^{2}$, and $\mathbb{E}_{\mathcal{D}_{n,t},\mathbf{r}_{n,t}}\left\|\nabla\widehat{\mathbf{F}}_{n,t}^{\texttt{P}}-\widehat{\mathbf{g}}_{n,t}^{\texttt{P}}\right\|^2\leq\sigma^{2}$. 
By Eq.~\eqref{222} in Assumption~\ref{ass2}, Jensen’s inequality and Eq.~\eqref{23} in Assumption~\ref{ass3}, we can get 
$\mathbb{E}_{\mathbf{r}_{n,t}}\left\| \nabla\widehat{\mathbf{F}}_{n,t}^{\texttt{B}}\right\|^2=\mathbb{E}_{\mathbf{r}_{n,t}}\left\|\mathbb{E}_{\mathcal{D}_{n,t}} \widehat{\mathbf{g}}_{n,t}^{\texttt{B}}\right\|^2\leq \mathbb{E}_{\mathcal{D}_{n,t},\mathbf{r}_{n,t}}\left\| \widehat{\mathbf{g}}_{n,t}^{\texttt{B}}\right\|^2\leq G^2$.  Therefore,
\begin{equation*}
\begin{split}
\kappa_1\left(\mathbf{r}_{n,t}\right)&\leq L_1\mathbb{E}_{\mathbf{r}_{n,t}}\left\|\mathbf{w}_{n,t}-\widehat{\mathbf{w}}_{n,t}\right\|\\&+\frac{L_2\eta^2-\eta}{2}G^2+\frac{3\eta\sigma^{2}}{2}.
\end{split}
\end{equation*}
From Jensen’s inequality, $\mathbb{E}_{\mathbf{r}_{n,t}}\left\|\mathbf{w}_{n,t}-\widehat{\mathbf{w}}_{n,t}\right\|=\mathbb{E}_{\mathbf{r}_{n,t}}\sqrt{\left\|\mathbf{w}_{n,t}-\widehat{\mathbf{w}}_{n,t}\right\|^2}\leq\sqrt{\mathbb{E}_{\mathbf{r}_{n,t}}\left\|\mathbf{w}_{n,t}-\widehat{\mathbf{w}}_{n,t}\right\|^2}$. And From Eq.~\eqref{36} in Lemma~\ref{lem1} and Eq.~\eqref{24} in Assumption~\ref{ass3}, it follows that

\begin{equation*}
\begin{split}
\kappa_1\left(\mathbf{r}_{n,t}\right)&\leq L_1M\sqrt{\left(1-r_{n,t}\right)}+\frac{L_2\eta^2-\eta}{2}G^2+\frac{3\eta\sigma^{2}}{2}\\&=\Psi_1(r_{n,t}).
\end{split}
\end{equation*}
\end{proof}
\begin{remark}
{\rm The Eq.~\eqref{pro1} represents the effect of model pruning. From the Proposition~\ref{proposition1}, the effect of model pruning is bounded. Although Eq.~\eqref{pro1} depends on $\mathbf{w}_{n,t}$, $\mathcal{D}_{n,t}$, and $\mathbf{r}_{n,t}$, we focus only on the effect of $\mathbf{r}_{n,t}$, and as shown in the proof that the other variables do not affect the bound $\Psi_1(r_{n,t})$. So we could define Eq.~\eqref{pro1} as $\kappa_1\left(\mathbf{r}_{n,t}\right)$.}
\end{remark}

Having quantified the effect of model pruning on the convergence process in Proposition~\ref{proposition1}, the natural next step is to address the other key term introduced in Proposition~\ref{proposition3}, namely, $\kappa_2\left(\mathbf{k}_{n,t}\right)$, which characterizes the influence of gradient sparsification. This step enables us to obtain a comprehensive understanding of the challenges posed by both compression techniques. Proposition~\ref{proposition2} is therefore devoted to analyzing this term and deriving an upper bound determined by the sparsification rate $k_{n,t}$.

\begin{proposition}
\label{proposition2}
{\rm For the given sparsification rate $k_{n,t}$, the effect of the gradient
sparsification of the base layers on the convergence $\kappa_2\left(\mathbf{k}_{n,t}\right)$ is upper-bounded by $\Psi_2(k_{n,t})$, with $\Psi_2(k_{n,t})$ given by}
\begin{equation}
\begin{split}
\Psi_2(k_{n,t})= \frac{\left(L_2\eta^2-3\eta\right)G^2}{2}\sum_{i = 1}^{N}\gamma_ik_{i,t}+\eta G^2.       
\end{split}
\end{equation}
\end{proposition}

\begin{proof}
From Jensen’s inequality, it follows that
\begin{equation*}
\begin{split}
\kappa_2\left(\mathbf{k}_{n,t}\right)&\leq\frac{L_2\eta^2-\eta}{2}\sum_{i = 1}^{N}\gamma_i\mathbb{E}\left\|{\widetilde{\mathbf{g}}_{i,t}^{\texttt{B}}}\right\|^2\\&+\eta\sum_{i = 1}^{N}\gamma_i\mathbb{E}\left\|\widehat{\mathbf{g}}_{i,t}^{\texttt{B}}-\widetilde{\mathbf{g}}_{i,t}^{\texttt{B}}\right\|^2.
\end{split}
\end{equation*}
From Eq.~\eqref{36}, Eq.~\eqref{37} in Lemma~\ref{lem1}, and Eq.~\eqref{23} in Assumption~\ref{ass3}, it follows that 
\begin{equation*}
\begin{split}
\kappa_2\left(\mathbf{k}_{n,t}\right)\leq \frac{\left(L_2\eta^2-3\eta\right)G^2}{2}\sum_{i = 1}^{N}\gamma_ik_{i,t}+\eta G^2=\Psi_2(k_{n,t}).
\end{split}
\end{equation*}
\end{proof}

\begin{remark}
{\rm The Eq.~\eqref{pro2} represents the effect of gradient sparsification. From the Proposition~\ref{proposition2}, the effect of gradient sparsification is bounded.  Although Eq.~\eqref{pro2} depends on $\mathbf{w}_{n,t}$, $\mathcal{D}_{n,t}$, $\mathbf{r}_{n,t}$, and $\mathbf{k}_{n,t}$, we focus only on the effect of $\mathbf{k}_{n,t}$, and as shown in the proof that the other variables do not affect the bound $\Psi_2(k_{n,t})$. So we could define Eq.~\eqref{pro2} as $\kappa_2\left(\mathbf{k}_{n,t}\right)$.}
\end{remark}

At this stage, all the preparatory steps have been completed. Proposition~\ref{proposition3} established the fundamental one-round convergence inequality, while Propositions~\ref{proposition1} and~\ref{proposition2} provided explicit upper bounds for the pruning- and sparsification-induced perturbation terms, respectively. With these components in place, we are now equipped to derive the final global convergence theorem. In Theorem~\ref{theorem:theorem1}, we integrate the results of all preceding propositions, apply a telescoping sum across $T$ iterations, and take the average, thereby obtaining the ultimate convergence guarantee for the proposed federated learning framework. This result clearly reveals the relationship between the convergence rate and the key system parameters.

\begin{theorem}
\label{theorem:theorem1}
{\rm Supposing $\eta\leq\frac{3}{L_2}$, the convergence of our proposed algorithm can be evaluated as follows}
\begin{equation}
\begin{split}
&\frac{1}{T}\sum_{t=1}^{T}\sum_{n = 1}^{N}\gamma_n\mathbb{E}\left\|\nabla\widehat{\mathbf{F}}_{n,t} \right\|^2 \\&
\leq \frac{2}{\eta T}\sum_{n = 1}^{N}\gamma_n\left(F_n\left( \mathbf{w}_{n,1} \right)-F_n\left( \mathbf{w}_{n,*}\right)\right)
\\&+\frac{2}{\eta T}\sum_{t=1}^{T}\sum_{n = 1}^{N}\gamma_n(\Psi_1(r_{n,t})+\Psi_2(k_{n,t})),
\end{split}
\label{26}
\end{equation}
{\rm where $T$ is training rounds of the algorithm, $N$ is the total number of clients participating in FL, $\mathbf{w}_{n,1}$ is the initial model parameter of client $n$, and $\mathbf{w}_{n,*}$ is the optimal model parameter for client $n$ that minimizes the empirical risk.}
\end{theorem}
\begin{proof}
From Eq.~\eqref{pro3} in Proposition~\ref{proposition3}, Proposition~\ref{proposition1}, and Proposition~\ref{proposition2}, rearranging some terms of the inequality, it follows that
\begin{equation*}
\begin{split}
&\mathbb{E}_{\mathbf{r}_{n,t}}\left\|\nabla\widehat{\mathbf{F}}_{n,t} \right\|^2 \leq \frac{2}{\eta}\left(F_n\left( \mathbf{w}_{n,t} \right)-\mathbb{E}\lbrack F_n\left( \mathbf{w}_{n,{t+1}}\right)\rbrack\right)
\\&+\frac{2}{\eta}(\Psi_1(r_{n,t})+\Psi_2(k_{n,t})).
\end{split}
\end{equation*}
A weighted sum of the empirical risk is performed after
$T$ training rounds across 
$N$ devices. The model parameter $\mathbf{w}_{n,t}$ in each training round depends on its previous state $\mathbf{w}_{n,t-1}$, which is influenced by the dataset 
 $\mathcal{D}_{n,t-1}$, the sparsification strategy  $\mathbf{k}_{n,t-1}$ and pruning strategy $\mathbf{r}_{n,t-1}$. Therefore, $\mathbf{w}_{n,t}$ also becomes a random variable.
Since $\nabla\widehat{\mathbf{F}}_{n,t}$ depends on both $\mathbf{w}_{n,t}$ and $\mathbf{r}_{n,t}$, the expectation is taken not only with respect to $\mathbf{r}_{n,t}$ but also the random factors 
 $\mathcal{D}_{n,t-1},\mathbf{k}_{n,t-1}$, and $\mathbf{r}_{n,t-1}$ which affect 
  $\mathbf{w}_{n,t}$. Consequently, the term  $\mathbb{E}_{\mathbf{r}_{n,t}}\left\|\nabla\widehat{\mathbf{F}}_{n,t} \right\|^2$ is generalized to $\mathbb{E}_{\mathcal{D}_{n,t-1},\mathbf{k}_{n,t-1},\mathbf{r}_{n,t-1},\mathbf{r}_{n,t}}\left\|\nabla\widehat{\mathbf{F}}_{n,t} \right\|^2$. 
To simplify notation without loss of clarity, the $\mathbb{E}\lbrack\cdot\rbrack$ 
 is resused to represent  $\mathbb{E}_{\mathcal{D}_{n,t-1},\mathbf{k}_{n,t-1},\mathbf{r}_{n,t-1},\mathbf{r}_{n,t}}\lbrack\cdot\rbrack$.
After $T$ rounds, the initial model parameter $\mathbf{w}_{n,1}$ is updated to $\mathbf{w}_{n,*}$. In the weighted sum of empirical risk,
the terms $\mathbb{E}_{\mathcal{D}_{n,t-1},\mathbf{k}_{n,t-1},\mathbf{r}_{n,t-1}}F_n\left( \mathbf{w}_{n,t} \right)$ between two consecutive iterations can cancel each other out.  Finally, it can be derived that
\begin{equation*}
\begin{split}
&\frac{1}{T}\sum_{t=1}^{T}\sum_{n = 1}^{N}\gamma_n\mathbb{E}\left\|\nabla\widehat{\mathbf{F}}_{n,t} \right\|^2 \\&
\leq \frac{2}{\eta T}\sum_{n = 1}^{N}\gamma_n\left(F_n\left( \mathbf{w}_{n,1} \right)-F_n\left( \mathbf{w}_{n,*}\right)\right)
\\&+\frac{2}{\eta T}\sum_{t=1}^{T}\sum_{n = 1}^{N}\gamma_n(\Psi_1(r_{n,t})+\Psi_2(k_{n,t})).
\end{split}
\end{equation*}
\end{proof}

\begin{remark}
{\rm Intuitively, the proof proceeds by bounding the per-round decrease of the global objective under $L$-smoothness, taking expectations to handle gradient noise and compression errors, and then summing the resulting inequalities across rounds. Under the stated Assumption~\ref{ass1} - \ref{ass3} and with the learning rate constrained as $\eta\le 3/L_2$, Theorem~\ref{theorem:theorem1} thus provides a rigorous finite-$T$ upper bound on the time-averaged gradient norm. A smaller right-hand side upper bound in Eq.~\eqref{26} implies faster convergence in the finite-round regime. Furthermore, by defining $
\overline{\Psi} \;=\; \frac{1}{T}\sum_{t=1}^T\sum_{n=1}^N \gamma_n\big(\Psi_1(r_{n,t})+\Psi_2(k_{n,t})\big),
$
we can relate the steady-state accuracy to the compression-induced errors. For sufficiently large $T$, we have 
$
\frac{1}{T}\sum_{t=1}^{T}\sum_{n=1}^{N}\gamma_n \mathbb{E}\big\|\nabla\widehat{F}_{n,t}\big\|^2  \le \frac{2}{\eta}\,\overline\Psi + O\!\Big(\frac{1}{T}\Big),
$
which implies that the steady-state (time-averaged) gradient magnitude scales on the order of $\sqrt{\frac{2}{\eta}\overline{\Psi}}$. Because the choices of $r_{n,t}$ and $k_{n,t}$ are produced by joint resource-compression optimization \eqref{27} without an explicit temporal decay, the algorithm is guaranteed to converge to an approximate stationary point whose accuracy depends on $\overline{\Psi}$ rather than to an exact stationary point. In practice, common engineering remedies such as learning-rate decay can substantially reduce the steady-state error without imposing additional client-side overhead when higher final accuracy is required.
}
\end{remark}

\begin{remark}
{\rm From the Theorem~\ref{theorem:theorem1}, the sparsification and pruning rates only impact the constant factor of the convergence, but the constant factor could significantly impact the convergence performance in practice, especially for commonly given iterations of the algorithm. A smaller constant factor indicates the algorithm converges faster for the same other factors, and optimizing this constant factor can be treated as a guide in designing the sparsification rates, pruning rates, and wireless bandwidth allocation, as discussed in \cite{9488839,9598845}.
}
\end{remark}
\begin{remark}
{\rm When $k_{n,t}$ or ${r_{n,t}}$ decreases, the constant factor of the convergence of our proposed algorithm increases, which implies the algorithm converges slower. Because sparsification and pruning introduce errors during training. The relative impact of sparsification and pruning on the constant factor depends on their corresponding coefficients.}
\end{remark}

\subsection{Problem Formulation}
To improve the algorithm's time efficiency, we can formulate a joint optimization problem to minimize the second terms on the right-hand side (RHS) of~\eqref{26} by optimizing the sparsification rates, pruning rates, and wireless bandwidth allocation. Specifically, the problem is defined as follows
\begin{subequations}
    \label{27}
    \begin{align}
    \min_{\substack{\{k_{n,t},r_{n,t},l_{n,t}\}}}
    &\sum_{t=1}^{T}\sum_{n = 1}^{N}\gamma_n\left(\theta_1\sqrt{\left(1-r_{n,t}\right)}-\theta_2 k_{n,t}\right),\\
    \label{27b}
    {\rm s.t.}~~~~~~&\tau_{n,t}^{\texttt{all}} \leq \tau^{\texttt{max}},\\
    &\label{27c}
    E_n^{\texttt{all}} \leq E_{n}^{\texttt{max}},\\
    &\label{27d}
    0 \leq k_{n,t} \leq 1,\\
    &\label{27e}
    0 \leq r_{n,t} \leq 1,\\
    &\label{27f}
    0 \leq l_{n,t} \leq 1,\\
    &\label{27g}
    \sum_{n=1}^N l_{n,t} = 1,
    \end{align}
\end{subequations}
where $\theta_1=\frac{2L_1 M}{\eta T} > 0$, $\theta_2=\frac{\left(3-L_2\eta\right)G^2}{T} > 0$. Constraints~\eqref{27b} and~\eqref{27c} indicate that the total latency of one training round and the total energy of each client are both limited. The~\eqref{27d},~\eqref{27e} and~\eqref{27f} define the domains of $k_{n,t}$, $r_{n,t}$ and $l_{n,t}$,~\eqref{27g} signifies that the sum of the wireless bandwidth allocation is 1. The three variables are interdependent, particularly due to the division operation between $k_{n,t}$ and $l_{n,t}$ in~\eqref{27b} and~\eqref{27c}, which makes the problem non-convex and difficult to solve. Fortunately, this problem can be reformulated as a difference of convex (DC) problem \cite{10.1007/978-3-319-06569-4_2,tan2019robust}, which can be efficiently solved using the Difference of Convex Functions Algorithm (DCA). 

\subsection{Algorithm Design}
To solve~\eqref{27}, we first apply variable substitution to transform the non-convex constraints in~\eqref{27b} and~\eqref{27c}. Then, we use the DCA to solve the problem.
Let the function be defined as
\begin{equation*}
\begin{split}
&\phi_1\left(k_{n,t}\right)=-k_{n,t}d^{\texttt{B}}\left(\log_2{\frac{1}{k_{n,t}}+FPP+1}\right).\\
\end{split}
\end{equation*}
The $\phi_1\left(k_{n,t}\right)$ is convex. By introducing a slack variable $z_{n,t}$ to replace the nonlinear terms in constraints~\eqref{27b} and~\eqref{27c}, the constraints can equivalently reformulated as
\begin{subequations}
    \label{28}
    \begin{align}
    &
    \label{28a}
    \tau_{n,t}^{\texttt{Comp}}+z_{n,t}\leq \tau^{\texttt{max}},\\
    &
    \label{28b}
    \sum_{t=1}^T\left( E_{n,t}^{\texttt{Comp}}+p_{n,t}z_{n,t}\right)\leq E_{n}^{\texttt{max}},\\
    &
    \label{28c}
    z_{n,t}=\frac{-\phi_1\left(k_{n,t}\right)}{l_{n,t}W\log_2\left(1+\frac{ h_{n,t} p_{n,t} }{N_0l_{n,t}W}\right)} .
    \end{align}
\end{subequations}
Constraints~\eqref{28a} and~\eqref{28b} are convex. By relaxing constraint~\eqref{28c}, the equivalent constraint is obtained as
\begin{subequations}
    \label{29}
    \begin{align}
    \label{29a}
    -z_{n,t}l_{n,t}W\log_2\left(1+\frac{ h_{n,t} p_{n,t} }{N_0l_{n,t}W}\right)\leq\phi_1\left(k_{n,t}\right).
    \end{align}
\end{subequations}
Before continuing with the next steps, it is necessary to explain that if a perspective function $\xi(x)$ is convex, then for a scalar $y>0$, the function $\mu(x,y) = y\xi(\frac{x}{y})$ is also convex. Let's introduce a slack variable $u_{n,t}$ and define the following function as
\begin{equation*}
\begin{split}
&\xi_1\left(l_{n,t}\right)= -l_{n,t}W\log_2\left(1+\frac{ h_{n,t} p_{n,t} }{N_0l_{n,t}W}\right),\\
&\lambda_1\left(u_{n,t},z_{n,t}\right)=-u_{n,t}W\log_2\left(1+\frac{ z_{n,t}h_{n,t} p_{n,t} }{u_{n,t}N_0W}\right).
\end{split}
\end{equation*}
 The $\xi_1\left(l_{n,t}\right)$ is convex, and if we replace $l_{n,t}$ with $\frac{u_{n,t}}{z_{n,t}}$, it can be regarded as a perspective function of the function $\lambda_1\left(u_{n,t},z_{n,t}\right)$, which means $\lambda_1\left(u_{n,t},z_{n,t}\right) = z_{n,t}\xi_1\left(\frac{u_{n,t}}{z_{n,t}}\right)$. And because the $\xi_1\left(l_{n,t}\right)$ is convex, the $\lambda_1\left(u_{n,t},z_{n,t}\right)$ is also convex.
 So constraint~\eqref{29a} can be equivalently transformed into
\begin{subequations}
    \label{30}
    \begin{align}
    &
    \label{30a}
   \lambda_1\left(u_{n,t},z_{n,t}\right) \leq \phi_1\left(k_{n,t}\right),\\
    &
    \label{30b}
    l_{n,t}=\frac{u_{n,t}}{z_{n,t}}.
    \end{align}
\end{subequations}
 Because the $\lambda_1\left(u_{n,t},z_{n,t}\right)$ is convex, the~\eqref{30a} can be expressed as the difference between two convex functions. After the above transformation, the optimization variable $l_{n,t}$ only appears in constraints~\eqref{27f},~\eqref{27g} and~\eqref{30b}, which means that $l_{n,t}$ can be removed from the problem~\eqref{27}. Constraints~\eqref{27f},~\eqref{27g}, and~\eqref{30b} can be equivalently transformed into
\begin{subequations}
    \label{31}
    \begin{align}
    &
    \label{31a}
    0\leq u_{n,t},\\
    &
    \label{31b}
    u_{n,t} \leq z_{n,t},\\
    &
    \label{31c}
    \sum_{n=1}^N \frac{u_{n,t}}{z_{n,t}} \leq 1.
    \end{align}
\end{subequations}
Constraints~\eqref{31a} and~\eqref{31b} are convex. But constraint~\eqref{31c} is still nonlinear. By introducing a slack variable $v_{n,t}$, constraint~\eqref{31c} can be equivalently transformed into
\begin{subequations}
    \label{32}
    \begin{align}
    &
    \label{32a}
    \sum_{n=1}^N \frac{v_{n,t}^2}{z_{n,t}} \leq 1,\\
    &\label{32b}
    u_{n,t} = v_{n,t}^2.
    \end{align}
\end{subequations}
Constraints~\eqref{32a} is convex. By relaxing constraint~\eqref{32b}, the equivalent constraint is obtained as
\begin{subequations}
    \label{33}
    \begin{align}
    \label{33a}
    u_{n,t} \leq v_{n,t}^2,
    \end{align}
\end{subequations}
Thus, the original problem~\eqref{27} can be equivalently transformed into
\begin{subequations}
    \label{34}
    \begin{align}
    \label{34a}
    \min_{\substack{\{k_{n,t},r_{n,t},\\z_{n,t},u_{n,t},v_{n,t}\}}}&\sum_{t=1}^{T}\sum_{n = 1}^{N}\gamma_n\left(\theta_1\sqrt{\left(1-r_{n,t}\right)}-\theta_2 k_{n,t}\right),\\
    {\rm s.t.}~~~~~~&{\rm constraints~in~}\eqref{27d},\eqref{27e},\eqref{28a},\eqref{28b},\eqref{30a},\nonumber\\ &~~~~~~~~~~~~~~~~~~\eqref{31a},\eqref{31b},\eqref{32a},\eqref{33a}.\nonumber
    \end{align}
\end{subequations}

Using DCA, the objective function~\eqref{34a} and the non-convex constraints~\eqref{30a} and~\eqref{33a} can be linearized \cite{10.1007/978-3-319-06569-4_2} and iteratively solve problem~\eqref{34}. Let the function be defined as
\begin{equation*}
\begin{split}
&\lambda_2\left(k_{n,t}\right)=-\theta_2\gamma_n k_{n,t},\\
&\phi_2\left(r_{n,t}\right)=-\theta_1\gamma_n\sqrt{\left(1-r_{n,t}\right)},\\
\end{split}
\end{equation*}
The $\lambda_2\left(k_{n,t}\right)$ and $\phi_2\left(r_{n,t}\right)$ are convex.
Suppose that the $(i-1)$-th iteration yields $k_{n,t}^{i},r_{n,t}^{i},z_{n,t}^{i},u_{n,t}^{i},v_{n,t}^{i}$. The $i$-th iteration requires solving the following convex problem as
\begin{subequations}
    \label{35}
    \begin{align}
    \min_{\substack{\{k_{n,t},r_{n,t},\\z_{n,t},u_{n,t},v_{n,t}\}}}&\sum_{t=1}^{T}\sum_{n = 1}^{N}\left\{\lambda_2\left(k_{n,t}\right)-\phi_2\left(r_{n,t}^{i}\right)\right\}\nonumber\\&-\sum_{t=1}^{T}\sum_{n = 1}^{N}\left\{\nabla\phi_2\left(r_{n,t}^{i}\right)\left(r_{n,t}-r_{n,t}^{i} \right)\right\},\\
    {\rm s.t.}~~~~~~&{\rm constraints~in~}\eqref{27d},\eqref{27e},\eqref{28a},\eqref{28b},\nonumber\\ &~~~~~~~~~~~~~~~~~~\eqref{31a},\eqref{31b},\eqref{32a},\nonumber\\
    &\lambda_1\left(u_{n,t},z_{n,t}\right)-\phi_1\left(k_{n,t}^{i}\right) \nonumber\\
    &-\nabla\phi_1\left(k_{n,t}^{i}\right)\left(k_{n,t}-k_{n,t}^{i} \right)\leq 0,\\
    &u_{n,t}-\left(v_{n,t}^{i}\right)^2-
    2v_{n,t}^{i}\left(v_{n,t}-v_{n,t}^{i} \right)\leq 0.
    \end{align}
\end{subequations}

Problem~\eqref{35} is convex and can be solved using standard optimization algorithms such as the interior-point method. Since the functions $\lambda_2(k_{n,t})$ and $\phi_2(r_{n,t})$ are differentiable, iteratively solving problem~\eqref{35} yields a local optimum for problem~\eqref{34} if the conditions $\nabla\lambda_2(k_{n,t})=0$ and $\nabla\phi_2(r_{n,t})=0$ are satisfied. The overall procedure is summarized in Algorithm~\ref{alg:alg2}.

Algorithm~\ref{alg:alg2} adopts the DCA framework, which decomposes the original non-convex problem~\eqref{27} into a sequence of convex subproblems~\eqref{35} with lower computational complexity. Benefiting from its fast convergence property, the algorithm achieves sufficient computational efficiency for practical applications, and no further design of lower-complexity methods is pursued. The overall computational complexity can be expressed as $O\left(I^{\texttt{DCA}} \cdot C^{\texttt{IP}}\right)$,
where $I^{\texttt{DCA}}$, namely $I$ in Algorithm~\ref{alg:alg2},  denotes the number of DCA iterations and $C_{\texttt{IP}}$ represents the complexity of solving one convex subproblem via the interior-point method. For convex subproblems, the interior-point method typically exhibits polynomial complexity, expressed as
$    C^{\texttt{IP}} = O\left(N^{3.5}\ln\left(\tfrac{1}{\epsilon^{\texttt{IP}}}\right)\right),$
with $\epsilon^{\texttt{IP}}$ denoting the solution accuracy of interior-point method.

\begin{algorithm}
\caption{DCA for problem~\eqref{34}}
\begin{algorithmic}
\STATE Initialize
$k_{n,t}^{0},r_{n,t}^{0},z_{n,t}^{0},u_{n,t}^{0},v_{n,t}^{0}$, set the number of iterations $I$ and the solution tolerance $\epsilon$
    \FOR{$i = 0,\ldots,I$}
        \STATE Compute $\nabla\phi_2\left(r_{n,t}^{i}\right)$ and $\nabla\phi_1\left(k_{n,t}^{i}\right)$
        \STATE Solve problem~\eqref{35} to obtain $k_{n,t}^{i+1},r_{n,t}^{i+1},z_{n,t}^{i+1},u_{n,t}^{i+1},v_{n,t}^{i+1}$
        \IF{$|k_{n,t}^{i+1}-k_{n,t}^{i}|\leq \epsilon^{\texttt{IP}}$ \AND $|r_{n,t}^{i+1}-r_{n,t}^{i}|\leq \epsilon^{\texttt{IP}}$ \AND $|z_{n,t}^{i+1}-z_{n,t}^{i}|\leq \epsilon^{\texttt{IP}}$ \AND $|u_{n,t}^{i+1}-u_{n,t}^{i}|\leq \epsilon^{\texttt{IP}}$ \AND
        $|v_{n,t}^{i+1}-v_{n,t}^{i}|\leq \epsilon^{\texttt{IP}}$
        }
        \RETURN $k_{n,t}^{i},r_{n,t}^{i},z_{n,t}^{i},u_{n,t}^{i},v_{n,t}^{i}$
        \ENDIF
    \ENDFOR
    \RETURN $k_{n,t}^{I},r_{n,t}^{I},z_{n,t}^{I},u_{n,t}^{I},v_{n,t}^{I}$
\end{algorithmic}
\label{alg:alg2}
\end{algorithm}

\section{Extensions and Future Works}

Our framework is built on parameter decoupling, partitioning the model into base  and personalization layers. In this paper we use a fixed partition ratio to isolate the coupling effects between gradient sparsification and model pruning. A key extension is to treat the partition ratio as an adaptive variable that balances knowledge sharing and personalization according to the learning environment.

Intuitively, the optimal partition ratio correlates with data heterogeneity: highly non-IID settings favor larger personalization capacity, in other words, smaller base fraction. We outline a two-stage roadmap for incorporating an adaptive partition ratio \(\rho_t^{\texttt{B}}\): a theoretical modeling stage and a practical adjustment stage.

In the theoretical stage, \(\rho_t^{\texttt{B}}\) can be introduced as an optimization variable via a binary indicator vector \({\mathbf{v}}_{t}^{\texttt{B}}\in\{0,1\}^d\) with \(\rho_t^{\texttt{B}}=\|{\mathbf{v}}_{t}^{\texttt{B}}\|_1/d\). Then \( {\mathbf{w}}_{n,t}^{\texttt{B}}={\mathbf{w}}_{n,t}\odot{\mathbf{v}}_{t}^{\texttt{B}}\) and \( {\mathbf{w}}_{n,t}^{\texttt{P}}={\mathbf{w}}_{n,t}\odot({\mathbf{1}}-{\mathbf{v}}_{t}^{\texttt{B}})\). Introducing \(\rho_t^{\texttt{B}}\) propagates through the analysis: the bounds in Theorem~\ref{theorem:theorem1} become functions of \(\rho_t^{\texttt{B}}\), i.e., \(\Psi_1(r_{n,t},\rho_t^{\texttt{B}})\), \(\Psi_2(k_{n,t},\rho_t^{\texttt{B}})\), and cost models, namely, transmission bits \(S_{n,t}\) and  computation latency \(\tau^{\texttt{Comp}}_{n,t}\), depend on \(d^{\texttt{B}}=\rho_t^{\texttt{B}}d\) and \(d^{\texttt{P}}=(1-\rho_t^{\texttt{B}})d\), leading to an extended joint optimization over \(\{k_{n,t},r_{n,t},\ell_{n,t},\rho_t^{\texttt{B}}\}\). This change propagates into the convergence bounds and cost models, making quantities such as the error bounds and communication/computation costs depend on \(\rho_t^{\texttt{B}}\).

In the adjustment stage, the theoretical reference for \(\rho_t^{\texttt{B}}\) can be refined online via an aggregation-gain heuristic: compare the improvement from aggregation versus purely local updates; increase \(\rho_t^{\texttt{B}}\) when aggregation gain is large, and decrease it when the gain is small. This two-stage approach thus combines upper-bound guidance with real-time performance feedback.

Although complete theoretical treatment and exhaustive experiments are beyond this paper's scope, adaptive partitioning is a promising direction to make the lightweight PFL framework more autonomous for heterogeneous and resource-constrained MEC scenarios. As complementary directions, we also note the value of structured or task-aware pruning, lightweight online scheduling and client selection, and integrating differential privacy or Byzantine-robust mechanisms; these directions require refined theoretical analysis and thorough empirical validation.

\section{Experiments}
\subsection{Experiment Settings}
\subsubsection{Simulation Parameters}
 We simulate a MEC scenario with heterogeneous communication and computation resources. Specifically, the channel gain between the MDs and the BS, the wireless transmission power of the MDs, and the CPU frequency of the MDs differ across various devices and training rounds. The BS coverage radius is 200 m, and the MDs are uniformly distributed in a circular region centered at the BS. The channel gain depends on path loss $128.1+37.6\log_{10}(\delta)$, where $\delta$ (km) is the distance between the MDs and the BS. The CPU frequency is uniformly distributed between 0.5 and 3.0 GHz, and the transmission power is uniformly distributed between 20 and 28 dBm. To ensure that differences in learning rate and batch size do not confound the results, these hyperparameters are fixed at a learning rate of 0.01 and a batch size of 32 throughout all experiments, unless otherwise specified. Some simulation parameters are set as shown in Table~\ref{tab:table2}.

\begin{table}[ht]
\caption{Simulation Parameters}
\centering
\begin{tabular}{c|c}
    \toprule
    Parameters & Values \\
    \midrule
    Path loss models & $128.1+37.6\log_{10}(\delta)$\\
    BS coverage radius & 200 m\\
    MDs CPU frequency & 0.5 $\sim$ 3.0 GHz\\
    Transmission power of BS & 33 dBm\\
    Transmission power of MDs & 20 $\sim$ 28 dBm\\
    AWGN noise power & -174 dBm/Hz\\
    Bandwidth of channel & 10 MHz\\
    Number of MDs, N & 100\\
    Number of training rounds, T &100\\
    Learning rate & 0.01\\
    batch size & 32\\
    \bottomrule
\end{tabular}
\label{tab:table2}
\end{table}

\subsubsection{Datasets and non-IID Partitions}
In our experiments, we use two image classification datasets (MNIST\cite{726791} and CIFAR-10\cite{He_2016_CVPR}), one real-world dataset (Street View House Numbers, SVHN\cite{netzer2011reading}), and one Natural Language Processing dataset (Microsoft Research Paraphrase Corpus, MRPC\cite{dolan2005automatically}).

To simulate statistical heterogeneity, we employ two different data partitioning methods. The first method assigns specific data classes to each client. For example, each client is randomly assigned to two data classes (denoted as Class-2), and the data from each class is then randomly assigned to the clients who own that class, following a uniform distribution. Fewer classes owned by a client indicate a higher degree of non-IID. The second method follows the approach in \cite{NEURIPS2020_18df51b9}, where the dataset is partitioned using a Dirichlet distribution, with the degree of non-IID controlled by the hyperparameter $\alpha$. A smaller value of $\alpha$ indicates a higher degree of non-IID. We set four different non-IID cases: (a) Class-4, (b) Class-2, (c) $\alpha$ = 0.5, and (d) $\alpha$ = 0.1. Each client has a test set, and the classes and distribution of the test set are the same as those of the local training dataset. Model evaluation is performed individually on each client because each client trains a different personalized model. The global test accuracy is calculated as the weighted average of the local test accuracies of each client, with the weights determined by $\gamma_n$. The global test accuracy is abbreviated as test accuracy.

\subsubsection{Baselines}
To evaluate the performance of our algorithm (FLPDSP-OPT), we chose six algorithms as baselines: FedAvg\cite{pmlr-v54-mcmahan17a}, FedPe\cite{arivazhagan2019federated}, LG-FedAvg\cite{liang2020think}, FedMask\cite{10.1145/3485730.3485929}, and two variants of FedAvg with global sparsification (FedAvg-S) and global pruning (FedAvg-P). FedPer and LG-FedAvg are two parameter decoupling algorithms: the former treats shallow layers as the base layers and deep layers as personalization layers, while the latter does the opposite. Neither of the two algorithms applies any sparsification or pruning methods. FedMask is a sparsification method that freezes model weights, applies masks to the deeper layers of the model, and trains the masks rather than the model itself in FL. The sparsification rate of FedAvg-S is $0.05$, and the pruning rate of FedAvg-P is $0.5$.

\subsubsection{Model Architectures} 
In our experiments, we employ three neural network models: LeNet-5 for MNIST, ResNet-18 for CIFAR-10 and SVHN, and BERT-Base with LoRA for MRPC.
The paper \cite{arivazhagan2019federated} points out that the reasonable decoupling of the model depends not only on the structure of the neural network but also on the specific training tasks and datasets. Taking the performance of the LeNet-5 for MNIST as an example, as shown in Table~\ref{tab:lenet5_results}, the proportion of base layers and personalization layers that achieves the optimal accuracy varies under different degrees of non-IID. When the degree of non-IID increases (i.e., when the hyperparameter $\alpha$ decreases), the proportion of base layers corresponding to the optimal accuracy becomes smaller, while the proportion of personalization layers becomes larger. To reasonably control variables and better investigate the impact of sparsification rates, pruning rates, and bandwidth allocation on the time efficiency and accuracy of FL systems under heterogeneous data and device conditions, a fixed model partitioning is adopted in this paper. The dynamic partitioning of base and personalization layers will be left as an important direction for future work.
\begin{table*}[htbp]
\centering
\caption{Rate of base and personalization layers (LeNet-5 for MNIST)}
\label{tab:lenet5_results}
\begin{tabular}{c|c|c|c|c|c|c}
\hline
Degree of non-IID ($\alpha$) & Base & Pers. & Base (\%) & Pers. (\%) & Time (s) & Acc \\
\hline
0.1 & Conv1 & Conv2 Conv3 FC1 FC2 & 0.25\% & 99.75\% & 2784 & 0.9938 \\
0.3 & Conv1 Conv2 & Conv3 FC1 FC2 & 4.17\% & 95.83\% & 4223 & 0.9891 \\
0.5 & Conv1 Conv2 Conv3 & FC1 FC2 & 82.15\% & 17.85\% & 4918 & 0.9859 \\
0.7 & Conv1 Conv2 Conv3 & FC1 FC2 & 82.15\% & 17.85\% & 5541 & 0.9826 \\
1.0 & Conv1 Conv2 Conv3 FC1 & FC2 & 98.62\% & 1.38\% & 5866 & 0.9785 \\
2.0 & Conv1 Conv2 Conv3 FC1 & FC2 & 98.62\% & 1.38\% & 5979 & 0.9668 \\
\hline
\end{tabular}
\end{table*}
Through experiments, we obtained reasonable settings for the above training tasks and datasets. For LeNet-5, convolutional layers and max-pooling layers are used as base layers, while fully connected layers are used as personalization layers. For ResNet-18 and SVHN, the last BasicBlock, along with the subsequent average pooling layers and fully connected layers, are used as personalization layers, with the rest layers as base layers. For BERT-Base with LoRA, the last four layers of the LoRA low-rank matrices together with the downstream classification layer are used as personalization layers, while the remaining layers are used as base layers.

\subsection{Convergence and Accuracy}
\subsubsection{Image Classification}
Firstly, the convergence and accuracy of these algorithms are compared on classical image datasets. Fig.~\ref{fig_1} to Fig.~\ref{fig_4} show these algorithms' performance regarding the loss function values and test accuracy over total training latency across different non-IID cases and datasets. On MNIST, the proposed algorithm consistently achieves the best time efficiency and test accuracy, and its advantage becomes more pronounced as the degree of non-IID increases. Although FedMask converges relatively fast, its accuracy is lower, with the gap being especially evident when the non-IID degree is small. FedPer and LG-FedAvg achieve final  the accuracy comparable to the proposed algorithm but require longer latency to converge, resulting in inferior time efficiency. Applying global sparsification to FedAvg (FedAvg-S) brings some improvements, whereas global pruning (FedAvg-P) performs even worse than FedAvg. Similar trends can be observed on CIFAR-10, the proposed algorithm still maintains the fastest convergence and highest test accuracy, while FedAvg and its global sparsification or pruning variants show no improvement. Overall, the proposed algorithm achieves the best time efficiency and test accuracy, as personalized models can better handle complicated non-IID datasets, while the jointly optimized sparsification and pruning rates minimize their impact on the convergence constant factor of the algorithm.
\begin{figure}[htbp]
\centering
\subfloat[]{
		\includegraphics[scale=0.5]{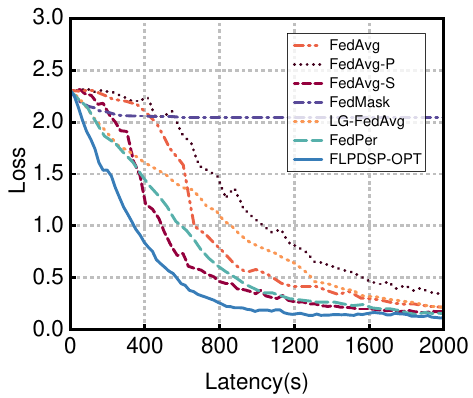}}
\subfloat[]{
		\includegraphics[scale=0.5]{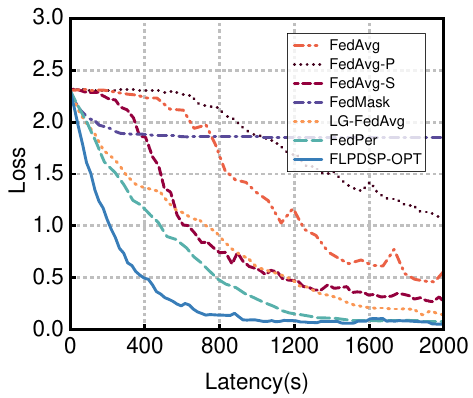}}
\\
\vspace{-10pt}
\subfloat[]{
		\includegraphics[scale=0.5]{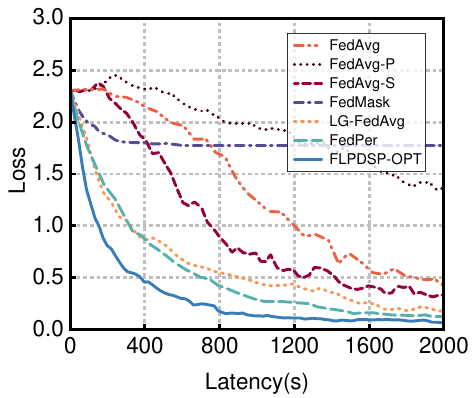}}
\subfloat[]{
		\includegraphics[scale=0.5]{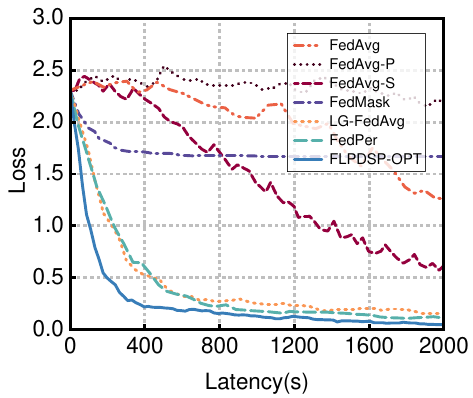}}
\caption{Loss vs. Latency on MNIST for different non-IID cases. (a) Class-4. (b) Class-2. (c) $\alpha$ = 0.5. (d) $\alpha$ = 0.1.}
\label{fig_1}
\end{figure}

\begin{figure}[htbp]
\centering
\subfloat[]{
		\includegraphics[scale=0.5]{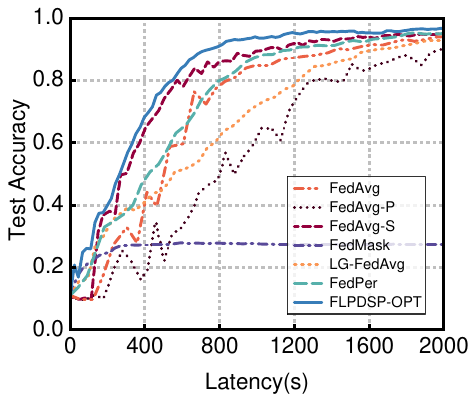}}
\subfloat[]{
		\includegraphics[scale=0.5]{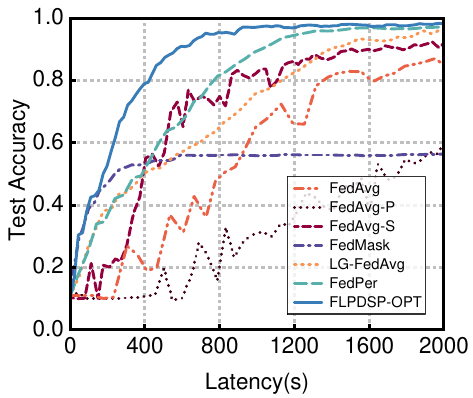}}
\\
\vspace{-10pt}
\subfloat[]{
		\includegraphics[scale=0.5]{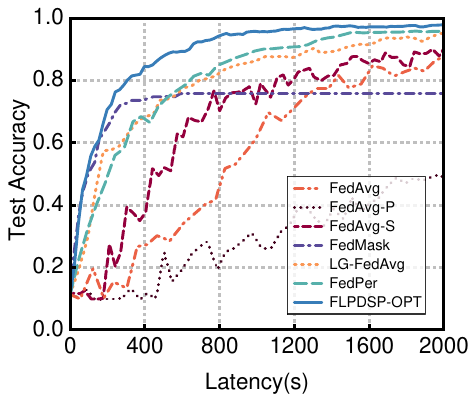}}
\subfloat[]{
		\includegraphics[scale=0.5]{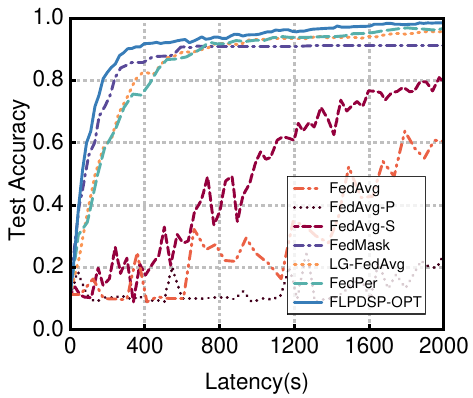}}
\caption{Test Accuracy vs. Latency on MNIST for different non-IID cases. (a) Class-4. (b) Class-2. (c) $\alpha$ = 0.5. (d) $\alpha$ = 0.1.}
\label{fig_2}
\end{figure}

\begin{figure}[htbp]
\centering
\subfloat[]{
		\includegraphics[scale=0.5]{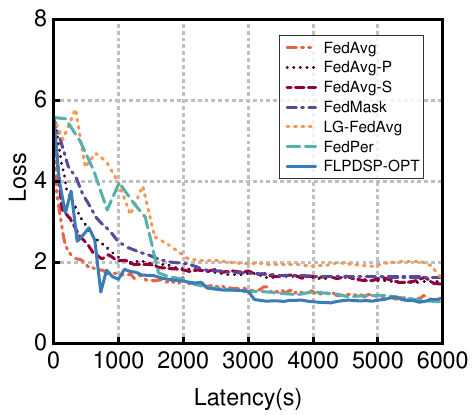}}
\subfloat[]{
		\includegraphics[scale=0.5]{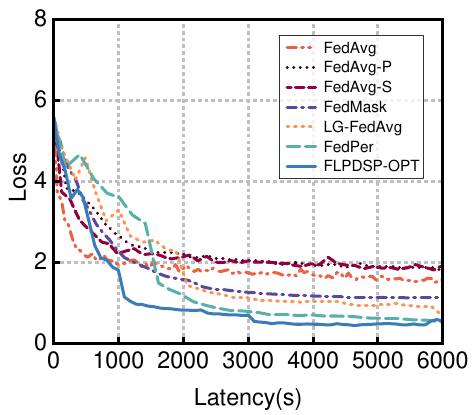}}
\\
\vspace{-10pt}
\subfloat[]{
		\includegraphics[scale=0.5]{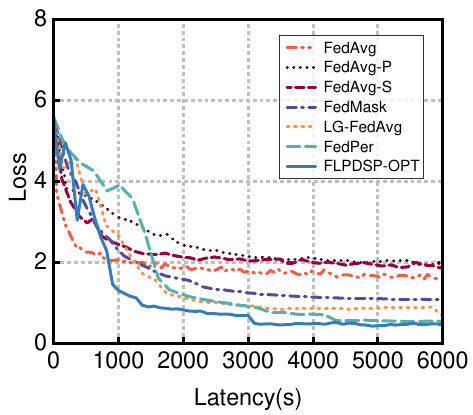}}
\subfloat[]{
		\includegraphics[scale=0.5]{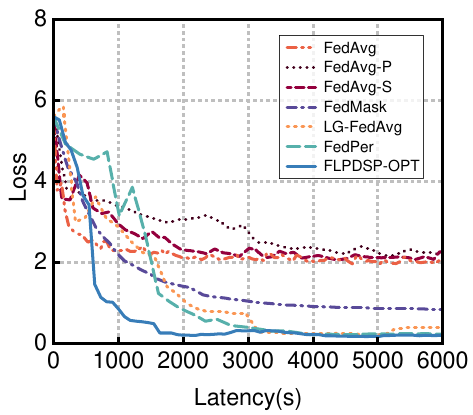}}
\caption{Loss vs. Latency on CIFAR-10 for different non-IID cases. (a) Class-4. (b) Class-2. (c) $\alpha$ = 0.5. (d) $\alpha$ = 0.1.}
\label{fig_3}
\end{figure}

\begin{figure}[htbp]
\centering
\subfloat[]{
		\includegraphics[scale=0.5]{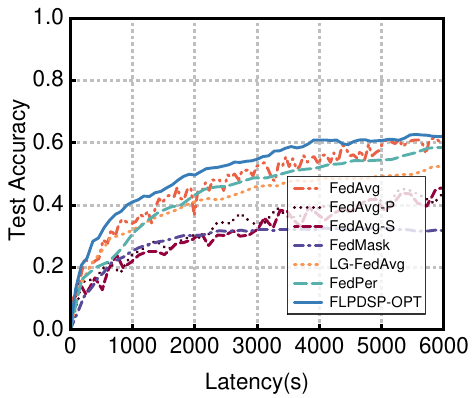}}
\subfloat[]{
		\includegraphics[scale=0.5]{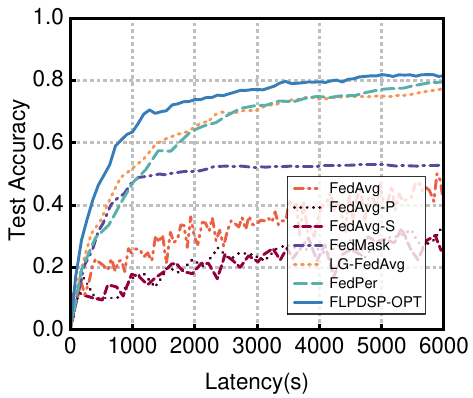}}
\\
\vspace{-10pt}
\subfloat[]{
		\includegraphics[scale=0.5]{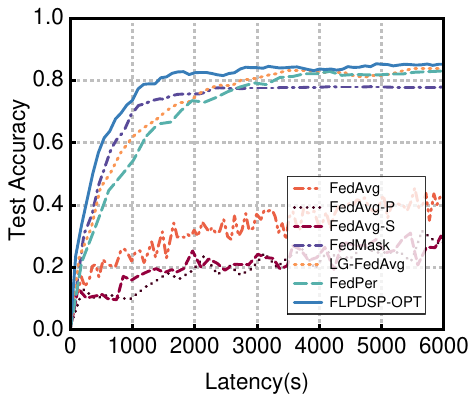}}
\subfloat[]{
		\includegraphics[scale=0.5]{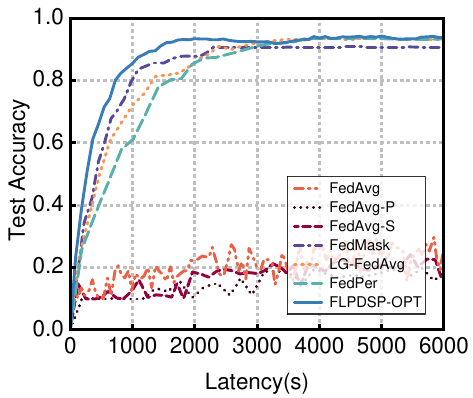}}
\caption{Test Accuracy vs. Latency on CIFAR-10 for different non-IID cases. (a) Class-4. (b) Class-2. (c) $\alpha$ = 0.5. (d) $\alpha$ = 0.1.}
\label{fig_4}
\end{figure}

\subsubsection{Natural Language Processing}
In further experiments, the algorithms was applied to natural language processing tasks. Fig.~\ref{fig_lora_loss} to Fig.~\ref{fig_lora_F1} present the variations of loss, accuracy, and F1 score of the algorithms over training latency. Since the MRPC dataset is a binary classification task, class-based non-IID partitioning was not adopted. The results show that in this more complex task, some personalized methods that performed well on simpler datasets (such as LG-FedAvg and FedMask) fail to maintain their advantages, with final test accuracy even falling below that of FedAvg. By contrast, our algorithm achieves faster convergence while maintaining high accuracy, thereby exhibiting a significant advantage in time efficiency.
\begin{figure}[htbp]
\centering
\subfloat[]{
		\includegraphics[scale=0.5]{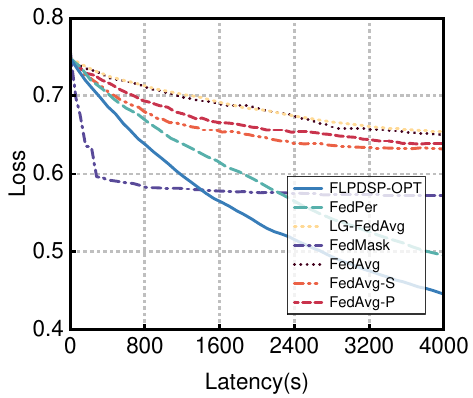}}
\subfloat[]{
		\includegraphics[scale=0.5]{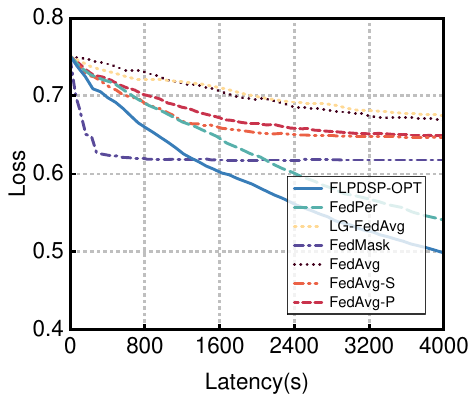}}
\caption{Loss vs. Latency on MRPC for BERT-Base with LoRA (a) $\alpha=1.0$ (b) $\alpha=2.0$}
\label{fig_lora_loss}
\end{figure}

\begin{figure}[htbp]
\centering
\subfloat[]{
		\includegraphics[scale=0.5]{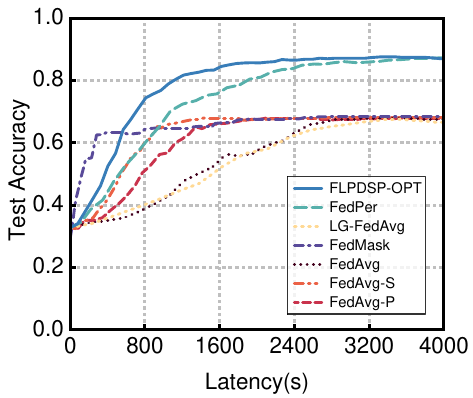}}
\subfloat[]{
		\includegraphics[scale=0.5]{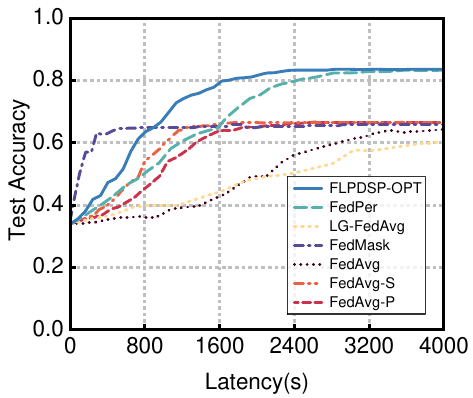}}
\caption{Test Accuracy vs. Latency on MRPC for BERT-Base with LoRA (a) $\alpha=1.0$ (b) $\alpha=2.0$}
\label{fig_lora_acc}
\end{figure}

\begin{figure}[htbp]
\centering
\subfloat[]{
		\includegraphics[scale=0.5]{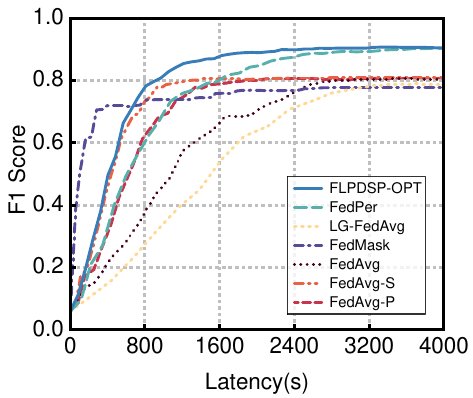}}
\subfloat[]{
		\includegraphics[scale=0.5]{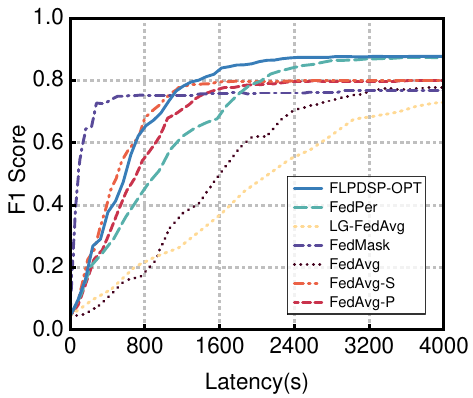}}
\caption{F1 Score vs. Latency on MRPC for BERT-Base with LoRA (a) $\alpha=1.0$ (b) $\alpha=2.0$}
\label{fig_lora_F1}
\end{figure}

\subsection{Communication and Computation Costs}
Table~\ref{tab:my-table} and Table~\ref{tab:my-table2} provide the overall communication and computation costs for all devices when FL training reaches a certain target accuracy (denoted as ToA@X) for different non-IID cases and datasets. It includes the total communication and computation latency, the total energy consumption, the total number of bits in uplink transmission of the sparse gradients of the base layers, and the total number of floating-point operations (FLOPs)\footnote{The FLOPs per CPU cycle depend on various factors, including the CPU's micro-architecture and instruction set. A higher value indicates greater computational power of the CPU. Since the FLOPs per cycle do not affect the conclusions of this paper, we set the FLOPs per cycle to a conservative value of 2\cite{dolbeau2018theoretical}.} for local computation of the base and personalization layers. Compared to the baselines, our algorithm reduces the total latency by around 30\% to 50\% when reaching the target accuracy, and also shows significant improvements in other costs. Moreover, as the degree of non-IID increases, the advantages of our algorithm in terms of the costs become more evident. Therefore, our algorithm not only achieves good performance on non-IID data but also effectively reduces the overall communication and computation costs of FL training.
\begin{table*}[htbp]
\centering
  \caption{communication and computation costs for MNIST}
    \begin{tabular}{lcrrrr|rrrr}
    \toprule
          &       & \multicolumn{1}{l}{Class-4} & \multicolumn{1}{l}{Class-2} & \multicolumn{1}{l}{$\alpha=0.5$} & \multicolumn{1}{l}{$\alpha=0.1$} & \multicolumn{1}{l}{Class-4} & \multicolumn{1}{l}{Class-2} & \multicolumn{1}{l}{$\alpha=0.5$} & \multicolumn{1}{l}{$\alpha=0.1$} \\
          & \multicolumn{1}{l}{ToA@X} & \multicolumn{4}{c}{Latency (s)} & \multicolumn{4}{c}{Energy (J)} \\
    \midrule
    FedAvg & \multirow{4}[2]{*}{80\%} & 894  & 1682  & 1447  & 2701  & 1389  & 2554  & 2177  & 4262 \\
    LG-FedAvg &       & 1197  & 1026  & 736   & 446   & 1945  & 1666  & 1233  & 773 \\
    FedPer &       & 912   & 867   & 713   & 475   & 1343 & 1269  & 1071   & 734 \\
    FLPDSP-OPT &       & \textbf{484} & \textbf{418} & \textbf{308} & \textbf{223} & \textbf{1336}  & \textbf{1173} & \textbf{871} & \textbf{645} \\
    \midrule
    FedAvg & \multirow{4}[2]{*}{90\%} & 1392  & 2594  & 2158  & 4260  & 2094  & 4067  & 3288  & 6676 \\
    LG-FedAvg &       & 1665  & 1456  & 1324  & 821   & 2684  & 2342  & 2146  & 1341 \\
    FedPer &       & 1243  & 1145  & 1009  & 955   & \textbf{1823} & 1687  & 1461  & 1403 \\
    FLPDSP-OPT &       & \textbf{770} & \textbf{572} & \textbf{506} & \textbf{441} & 2117  & \textbf{1559} & \textbf{1391} & \textbf{1215} \\
    \midrule
          & \multicolumn{1}{l}{ToA@X} & \multicolumn{4}{c}{Transmission bits ($\times10^9$)} & \multicolumn{4}{c}{TFLOPs} \\
    \midrule
    FedAvg & \multirow{4}[2]{*}{80\%} & 0.71  & 1.32  & 1.12  & 2.15  & \textbf{12.54} & 23.41  & 20.06 & 38.11 \\
    LG-FedAvg &       & 1.03  & 0.89  & 0.64  & 0.39  & 24.22 & 20.88 & 15.14  & 11.18 \\
    FedPer &       & 0.88  & 0.84  & 0.72  & 0.46  & 15.88  & \textbf{15.04} & 12.55  & 8.34 \\
    FLPDSP-OPT &       & \textbf{0.68} & \textbf{0.60} & \textbf{0.44} & \textbf{0.32} & 17.82  & 15.42 & \textbf{11.37} & \textbf{8.02} \\
    \midrule
    FedAvg & \multirow{4}[2]{*}{90\%} & 1.17  & 2.05  & 1.68  & 3.36  & \textbf{19.22} & 36.76 & 30.08 & 60.16 \\
    LG-FedAvg &       & 1.42  & 1.25  & 1.14  & 0.71  & 33.42 & 29.24 & 26.54 & 20.72 \\
    FedPer &       & 1.21  & 1.12  & 0.98  & 0.93  & 21.72 & \textbf{20.06} & 21.54 & 18.72 \\
    FLPDSP-OPT &       & \textbf{1.10} & \textbf{0.83} & \textbf{0.71} & \textbf{0.62} & 28.34 & 21.02 & \textbf{18.64} & \textbf{16.22} \\
    \bottomrule
    \end{tabular}%
  \label{tab:my-table}%
\end{table*}%

\subsection{Real-World MEC Environments}
To further evaluate the effectiveness and robustness of the proposed algorithm in Real-World MEC Environments, we conduct experiments on the SVHN dataset under different bandwidth conditions and random bandwidth fluctuations. The SVHN dataset, derived from real-world house number recognition, better reflects practical MEC applications such as smart logistics, where drones or autonomous vehicles must recognize house numbers for accurate delivery.

As shown in Table~\ref{tab:average_k_r} and Fig.~\ref{fig_5}, when bandwidth is small, communication becomes the main bottleneck, sparsification significantly improves time efficiency, while pruning shows limited effect. When bandwidth is larger, pruning becomes more beneficial by reducing local computation costs, while sparsification plays a diminished role. This highlights the coupling among sparsification, pruning, and bandwidth allocation. Compared with fixed-ratio settings, joint optimization adaptively balances sparsification and pruning according to bandwidth conditions and device heterogeneity, while flexibly allocating communication resources.

In scenarios with random bandwidth fluctuations, the algorithm also adapts sparsification, pruning, and bandwidth allocation to maintain high time efficiency and test accuracy. These results demonstrate the necessity and superiority of the joint optimization, as well as its robustness in complex MEC environments.
\begin{table*}[htbp]
  \centering
  \caption{communication and computation costs for CIFAR-10}
    \begin{tabular}{cccc|cccc|ccc}
    \toprule
          &       & \multicolumn{2}{c|}{Class-4} &       & Class-2 & $\alpha=0.5$ & \multicolumn{1}{c}{$\alpha=0.1$} & Class-2 & $\alpha=0.5$ & $\alpha=0.1$ \\
          & ToA@X & Latency (s) & Energy (J) & ToA@X & \multicolumn{3}{c}{Latency (s)} & \multicolumn{3}{c}{Energy (J)} \\
    \midrule
    LG-FedAvg & \multirow{3}[2]{*}{60\%} & 12091 & 19919 & \multirow{3}[2]{*}{70\%} & 2959  & 1737  & 1036  & 4958  & 2865  & 1809 \\
    FedPer &       & 7787  & 13198 &       & 2979  & 1797  & 1391  & 4850  & 2964  & 2309 \\
    FLPDSP-OPT &       & \textbf{3912} & \textbf{10926} &       & \textbf{1546} & \textbf{819} & \textbf{546} & \textbf{4396} & \textbf{2312} & \textbf{1647} \\
    \midrule
    LG-FedAvg & \multirow{3}[2]{*}{70\%} & 25011 & 41884 & \multirow{3}[2]{*}{80\%} & 10682 & 3512  & 1379  & 17507 & 5795  & 2329 \\
    FedPer &       & 20141 & \textbf{34244} &       & 7601  & 3394  & 1797  & \textbf{12826} & 5518  & 2964 \\
    FLPDSP-OPT &       & \textbf{14188} & 40489 &       & \textbf{5004} & \textbf{1546} & \textbf{728} & 14122 & \textbf{4396} & \textbf{2094} \\
    \midrule
          & ToA@X & Transmission bits ($\times10^9$) & TFLOPs & ToA@X & \multicolumn{3}{c}{Transmission bits ($\times10^9$)} & \multicolumn{3}{c}{TFLOPs} \\
    \midrule
    LG-FedAvg & \multirow{3}[2]{*}{60\%} & 245.65 & 227.38 & \multirow{3}[2]{*}{70\%} & 61.42  & 36.11 & 21.67 & 56.84 & 33.42 & 25.06 \\
    FedPer &       & 196.03 & \textbf{133.72} &       & 73.53  & 44.10  & 34.31 & \textbf{50.14} & 31.18 & 23.42 \\
    FLPDSP-OPT &       & \textbf{134.89} & 140.14 &       & \textbf{54.91} & \textbf{28.51} & \textbf{19.61} & 55.62 & \textbf{29.62} & \textbf{19.82} \\
    \midrule
    LG-FedAvg & \multirow{3}[2]{*}{70\%} & 512.88 & 474.72 & \multirow{3}[2]{*}{80\%} & 216.71 & 72.23 & 28.89 & 201.58 & 66.87 & 30.74 \\
    FedPer &       & 504.72 & \textbf{344.32} &       & 191.11 & 83.32 & 44.11 & \textbf{130.38} & 56.85 & 29.88 \\
    FLPDSP-OPT &       & \textbf{487.61} & 508.06 &       & \textbf{173.37} & \textbf{54.91} & \textbf{26.03} & 178.86 & \textbf{54.65} & \textbf{26.36} \\
    \bottomrule
    \end{tabular}%
\label{tab:my-table2}%
\end{table*}%

\begin{table}[htbp]
\centering
\caption{average sparsification and pruning rate under different bandwidth}
\begin{tabular}{lcc}
\hline
\textbf{Bandwidth} & $\bar{k}$ & $\bar{r}$ \\
\hline
10 MHz              & 0.74 & 0.82 \\
30 MHz              & 0.86 & 0.91 \\
50 MHz              & 0.98 & 0.96 \\
10$\sim$50 MHz      & 0.78 & 0.97 \\
20$\sim$40 MHz      & 0.66 & 0.94 \\
\hline
\end{tabular}
\label{tab:average_k_r}
\end{table}

\begin{figure}[htbp]
\centering
\subfloat[]{
		\includegraphics[scale=0.5]{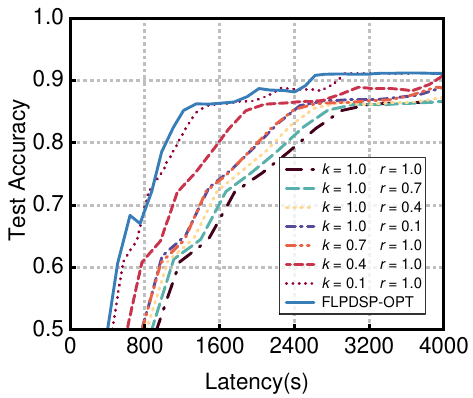}}
\subfloat[]{
		\includegraphics[scale=0.5]{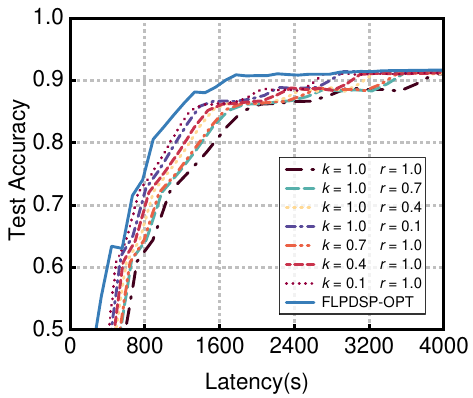}}
\\
\vspace{-10pt}
\subfloat[]{
		\includegraphics[scale=0.5]{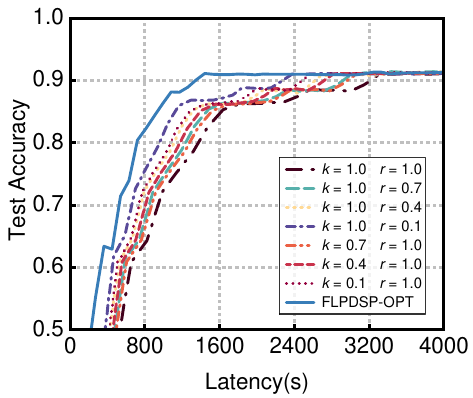}}
\subfloat[]{
		\includegraphics[scale=0.5]{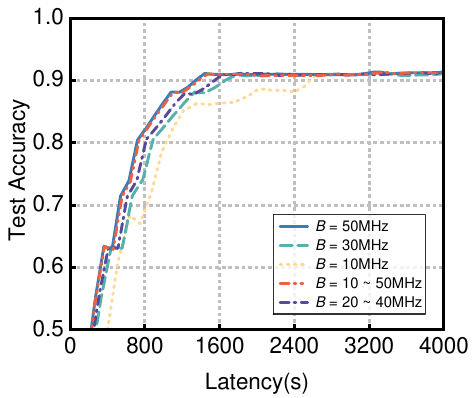}}
\caption{Test Accuracy vs. Latency on SVHN for different bandwidth cases. (a) bandwidth = 10 MHz  (b) bandwidth = 30 MHz (c) bandwidth = 50 MHz (d) dynamic bandwidth}
\label{fig_5}
\end{figure}

\subsection{Ablation Studies}
Fig.~\ref{fig_sp} analyzes the contributions of different module combinations (parameter decoupling, sparsification, and pruning) to the performance improvement of the algorithm. Specifically, FedAvg-P represents pruning only, FedAvg-S represents sparsification only, and FLPDSP ($k=1, r=1$) represents parameter decoupling only. Setting different values of $k$ and $r$ for FLPDSP corresponds to applying sparsification to the base layers, pruning to the personalization layers, or both simultaneously on top of parameter decoupling.

From Fig.~\ref{fig_sp} (a), it can be seen that when only one module is enabled, parameter decoupling contributes the most to improving test accuracy, sparsification provides moderate improvement, while pruning actually reduces test accuracy.

From Fig.~\ref{fig_sp} (b), for the combination of sparsification and parameter decoupling, decreasing the sparsification rate $k$ accelerates convergence, but the final test accuracy declines.

From Fig.~\ref{fig_sp} (c), for the combination of pruning and parameter decoupling, decreasing the pruning rate $r$ accelerates convergence in the early stage of training (before about 900s). However, as training progresses, the accuracy degradation caused by pruning becomes evident, and the smaller the pruning rate $r$, the more pronounced the decline in test accuracy.

From Fig.~\ref{fig_sp} (d), when parameter decoupling, pruning, and sparsification are combined simultaneously, performance improves further. In the early stage of training (before about 400s), both time efficiency and test accuracy even surpass joint optimization, but in the later stage (after about 400s), they fall below the results of joint optimization.

In summary, using only one of the three modules, or combining parameter decoupling with sparsification, or parameter decoupling with pruning, or even combining all three with fixed sparsification and pruning rates, cannot achieve faster convergence or higher test accuracy than joint optimization. This demonstrates that the optimization problem proposed in this work is both necessary and effective.

\begin{figure}[htbp]
\centering
\subfloat[]{
		\includegraphics[scale=0.5]{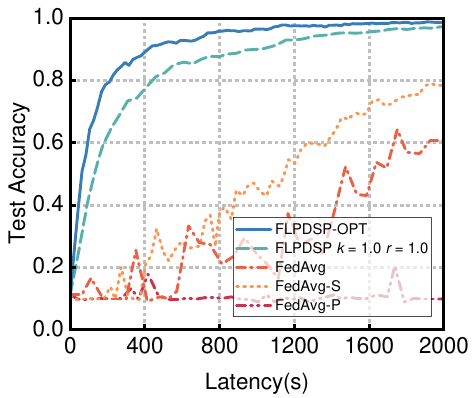}}
\subfloat[]{
		\includegraphics[scale=0.5]{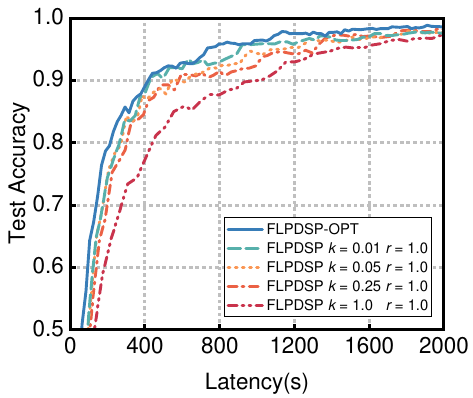}}
\\
\vspace{-10pt}
\subfloat[]{
		\includegraphics[scale=0.5]{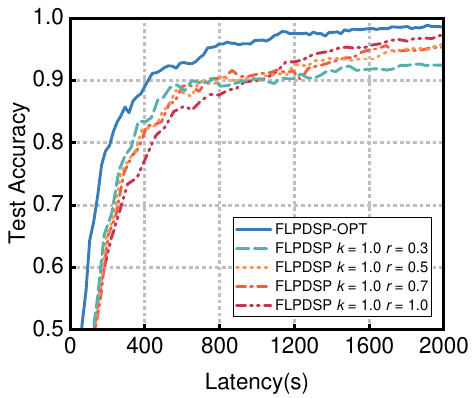}}
\subfloat[]{
		\includegraphics[scale=0.5]{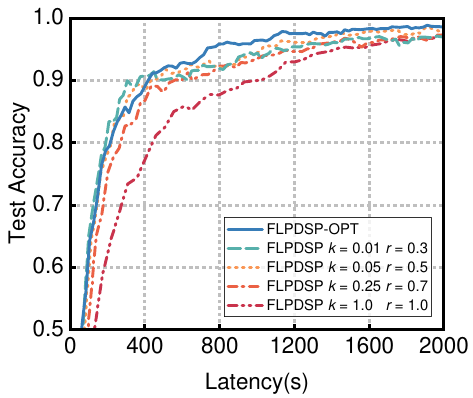}}
\caption{Test Accuracy vs. Latency on MNIST ($\alpha=0.1$). (a) only pruning, sparsification or parameter decoupling  (b) sparsification and parameter decoupling (c) pruning and parameter decoupling (d) pruning, sparsification and parameter decoupling}
\label{fig_sp}
\end{figure}

\section{Conclusions}
In this paper, we have introduced a novel lightweight PFL training approach that reduces communication and computation costs by addressing client heterogeneity based on parameter decoupling with gradient sparsification and model pruning. Furthermore, we have analyzed the convergence of our proposed algorithm, which lays the foundation for further optimization. To this end, a joint optimization problem has been formulated based on the constant factor of convergence, device heterogeneity, and limited MEC system resources. Moreover, the experiments demonstrate the superiority of our approach in improving the algorithm's time efficiency and reducing overall communication and computation costs, especially training latency, particularly in cases with high degrees of data heterogeneity. This advancement is expected to promote the application of FL in real-world MEC environments.

\bibliographystyle{IEEEtran}
\bibliography{referrences}

\begin{IEEEbiography}
[{\includegraphics[width=1in,height=1.25in,keepaspectratio]{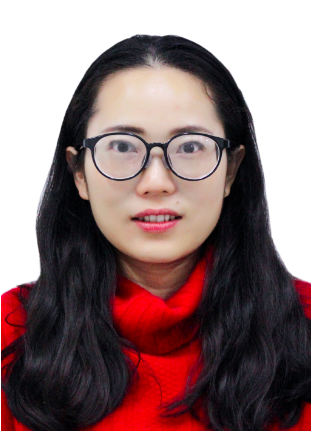}}]
{Jinghong Tan} received the B.E. degree in Communication Engineering from Shandong University, Jinan, China, in 2014, and the Ph.D. degree in Electronic Engineering from the Singapore University of Technology and Design (SUTD), Singapore, in 2019. She is currently a Lecturer with the School of Software, Yunnan University, Kunming, China. Her research interests include wireless edge computing, edge intelligence, and incentive mechanisms for edge networks.
\end{IEEEbiography}

 \begin{IEEEbiography}
  [{\includegraphics[width=1in,height=1.25in, keepaspectratio]{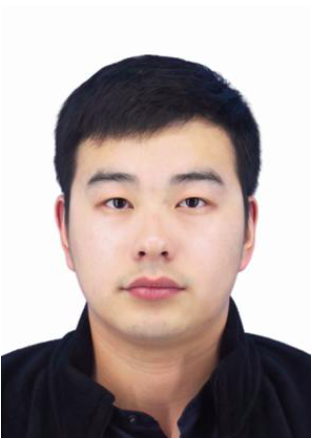}}]
  {Zhichen Zhang} received his M.E. degree in software engineering from Yunnan University, Kunming, China, in 2025. His research focuses on federated learning.
  \end{IEEEbiography}

 \begin{IEEEbiography}
[{\includegraphics[width=1in,height=1.25in,keepaspectratio]{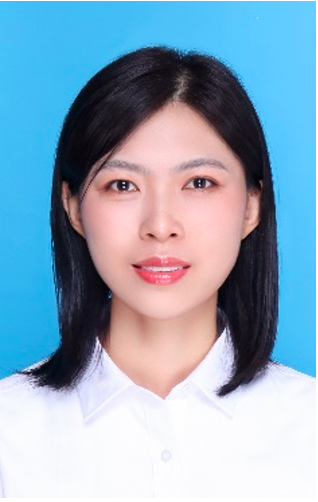}}]
{Kun Guo} (Member, IEEE) received the B.E. degree in Telecommunications Engineering from Xidian University, Xi'an, China, in 2012, where she received the Ph.D. degree in communication and information systems in 2019. From 2019 to 2021, she was a Post-Doctoral Research Fellow with the Singapore University of Technology and Design (SUTD), Singapore. Currently, she is a Research Professor with the School of Communications and Electronics Engineering at East China Normal University, Shanghai, China. Her research interests include wireless edge computing and intelligence, as well as non-terrestrial networks.

\end{IEEEbiography}

 \begin{IEEEbiography}[{\includegraphics[width=1in,height=1.25in,keepaspectratio]{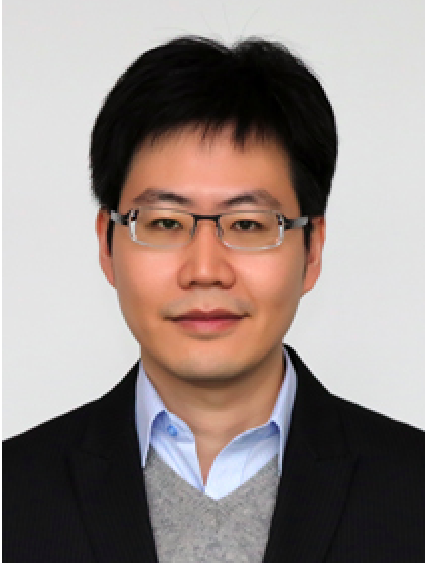}}]
{Tsung-Hui Chang} (Fellow, IEEE) received the B.S. degree in electrical engineering and the Ph.D. degree in communications engineering from the National Tsing Hua University (NTHU), Hsinchu, Taiwan, in 2003 and 2008, respectively. Currently, he is a Professor and the Associate Dean (Education) of the School of Artificial Intelligence, The Chinese University of Hong Kong, Shenzhen (CUHK-SZ), China, and the Associate Director of Guangdong Provincial Key Laboratory of Big Data Computing. Before joining CUHK-SZ, he was with the National Taiwan University of Science and Technology (NTUST), the University of California, Davis, as a Postdoctoral Researcher and a Faculty Member, respectively. His research interests include signal processing and optimization problems in data communications and machine learning. He has been an Elected Member of the IEEE Signal Processing Society (SPS) Signal Processing for Communications and Networking Technical Committee (SPCOM TC) since 2020. He received the Young Scholar Research Award of NTUST in 2014, the IEEE ComSoc Asian-Pacific Outstanding Young Researcher Award in 2015, the IEEE SPS Best Paper Awards in 2018 and 2021, the Outstanding Faculty Research Award of SSE of CUHK-SZ in 2021, and the Outstanding Research Award of CUHK-SZ in 2024. He is the Founding Chair of the IEEE SPS Integrated Sensing and Communication Technical Working Group (ISAC TWG) and the elected Regional Director-at-Large of Board of Governors of IEEE SPS from 2022 to 2023. He has served on the editorial board for major SP journals, including an Associate Editor for IEEE Transactions on Signal Processing from 2014 to 2018, IEEE Transactions on Signal and Information Processing Over Networks from 2015 to 2018, IEEE Open Journal of Signal Processing since 2020, and a Senior Area Editor for IEEE Transactions on Signal Processing from 2021 to 2025.

\end{IEEEbiography}

\begin{IEEEbiography}[{\includegraphics[width=1in,height=1.25in,keepaspectratio]{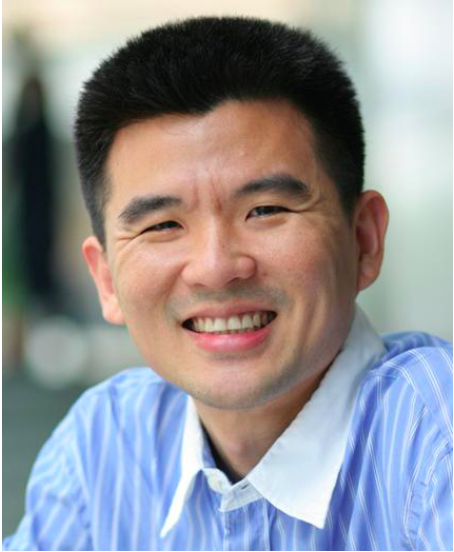}}]
{Tony Q.S. Quek}(Fellow, IEEE) received the B.E.\ and M.E.\ degrees in electrical and electronics engineering from the Tokyo Institute of Technology in 1998 and 2000, respectively, and the Ph.D.\ degree in electrical engineering and computer science from the Massachusetts Institute of Technology in 2008. Currently, he is the Associate Provost (AI \& Digital Innovation) and Cheng Tsang Man Chair Professor with Singapore University of Technology and Design (SUTD). He also serves as the Director of the Future Communications R\&D Programme, and the ST Engineering Distinguished Professor. He is a co-founder of Silence Laboratories and NeuroRAN. His current research topics include wireless communications and networking, network intelligence, non-terrestrial networks, open radio access network, AI-RAN, and 6G.

Dr.\ Quek was honored with the 2008 Philip Yeo Prize for Outstanding Achievement in Research, the 2012 IEEE William R. Bennett Prize, the 2015 SUTD Outstanding Education Awards -- Excellence in Research, the 2016 IEEE Signal Processing Society Young Author Best Paper Award, the 2017 CTTC Early Achievement Award, the 2017 IEEE ComSoc AP Outstanding Paper Award, the 2020 IEEE Communications Society Young Author Best Paper Award, the 2020 IEEE Stephen O. Rice Prize, the 2020 Nokia Visiting Professor, the 2022 IEEE Signal Processing Society Best Paper Award, the 2024 IIT Bombay International Award For Excellence in Research in Engineering and Technology, the IEEE Communications Society WTC Recognition Award 2024, and the Public Administration Medal (Bronze). He is an IEEE Fellow, a WWRF Fellow, an AIAA Fellow, and a Fellow of the Academy of Engineering Singapore.
\end{IEEEbiography}

\end{document}